\documentclass[11pt,a4paper]{article}
\usepackage[T1]{fontenc}
\usepackage[utf8]{inputenc}
\usepackage[british]{babel}
\usepackage{amsmath}
\usepackage{amsfonts}
\usepackage{amssymb}
\usepackage{amsthm}
\usepackage[margin=3.3cm]{geometry}
\usepackage[tt=false, type1=true]{libertine}
\usepackage[varqu]{zi4}
\usepackage[libertine]{newtxmath}
\usepackage[cal=euler,scr=boondox]{mathalfa}
\usepackage[square,numbers]{natbib}
\usepackage{bm}
\usepackage{dsfont}
\usepackage{enumitem}
\usepackage{authblk}
\usepackage{hyperref}
\usepackage{doi}

\theoremstyle{plain}
	\newtheorem{theorem}{Theorem}[section]
	\newtheorem{lemma}[theorem]{Lemma}
	\newtheorem{proposition}[theorem]{Proposition}
	\newtheorem{corollary}[theorem]{Corollary}
\theoremstyle{definition}
	\newtheorem{definition}[theorem]{Definition}	
\theoremstyle{remark}
	\newtheorem{remark}[theorem]{Remark}
	\newtheorem{example}[theorem]{Example}

\DeclareMathAlphabet{\mathbbold}{U}{bbold}{m}{n}
\newcommand*{\1}{\mathbbold{1}}	
\newcommand*{\0}{\mathbbold{0}}

\let\citep\cite
\newcommand{\topo}{\mathcal{T}}
\newcommand{\ann}{\mathcal{A}}
\newcommand{\base}{\mathcal{B}}
\newcommand{\subbase}{\mathcal{S}}
\newcommand{\powset}{\mathcal{P}}
\newcommand{\val}{\mathcal{V}}
\newcommand{\evid}{\mathcal{E}}
\newcommand{\Int}{\mathit{Int}}
\newcommand{\For}{\mathit{For}}
\newcommand{\Prop}{\mathit{Prop}}
\newcommand{\F}{\mathrm{F}}
\newcommand{\M}{\mathrm{M}}
\newcommand{\A}{\mathrm{A}}
\newcommand{\E}{\mathrm{E}}
\newcommand{\Bel}{\mathrm{B}}
\newcommand{\Belop}{\mathit{Bel}}
\newcommand{\Kn}{\mathrm{K}}
\newcommand{\Knop}{\mathit{Kn}}

\newcommand{\f}{\varphi}
\newcommand{\ff}{\psi}
\newcommand{\xto}{\xrightarrow}
\newcommand{\tot}{\leftrightarrow}
\newcommand{\sem}[1]{\llbracket #1 \rrbracket}
\newcommand{\semM}[1]{\llbracket #1 \rrbracket_{\bm{M}}}
\newcommand{\proves}{\mathrel{|}\joinrel\mkern-.5mu\mathrel{-}}		

\title{\bfseries Knowledge on a Budget}

\author[1]{Ondrej Majer}
\author[2]{Krishna Manoorkar}
\author[2]{Wolfgang Poiger}
\author[2]{Igor Sedl\'ar}
\affil[1]{Institute of Philosophy, Czech Academy of Sciences,\authorcr Prague, Czech Republic}
\affil[2]{Institute of Computer Science, Czech Academy of Sciences,\authorcr Prague, Czech Republic}

\date{ }

\begin{document}

\maketitle

\begin{abstract}
In various computational systems, accessing information incurs time, memory or energy costs. However, standard epistemic logics usually model the acquisition of evidence as a cost-free process, which restricts their applicability in environments with limited resources.  
In this paper, we bridge the gap between qualitative epistemic reasoning and quantitative resource constraints by introducing semiring-annotated topological spaces (seats). Building on Topological Evidence Logic (TEL), we extend the representation of evidence as open sets, adding an annotation function that maps evidence to semiring ideals, representing the resource budgets sufficient for observation. This framework allows us to reason not only about what is observable in principle, but also about what is affordable given a specific budget. 
We develop a family of seat-based epistemic logics with resource-indexed modalities and provide sound, strongly complete axiomatisations for these logics. Furthermore, we introduce suitable notions of bisimulation and disjoint union to delineate the expressive power of our framework.
\end{abstract}


\maketitle






\section{Introduction}

In typical computational systems, \emph{knowledge rarely comes for free}. Whether it is a database query engine, a security protocol or a robotic agent, accessing information consumes resources, whether that be time, memory, energy or money. However, standard epistemic logics usually treat the acquisition of evidence as a cost-free, idealized process. This discrepancy limits the applicability of logical models to real-world scenarios in which agents must operate within strict resource constraints. 
The central issue addressed in this article is how to integrate quantitative resource constraints into a qualitative logic of evidence and justification.

We build upon \emph{Topological Evidence Logic} (TEL), a robust framework that models epistemic concepts using the tools of topology; see \citep{BaltagEtAl2016,BaltagEtAl2019,
BaltagEtAl2022,BaltagEtAl2022a,BaltagEtAl2025a,BaltagEtAl2025,BjorndahlOzgun2019,
FernandezGonzalez2018,OzgunEtAl2025} for example. TEL represents evidence by open sets in a topological space. This approach is rooted in the topological semantics for modal logic \cite{McKinseyTarski1944} and is consistent with the use of topology to model observable properties in the domain-theoretic foundations of programming semantics \cite{Abramsky1991,Smyth1992,Vickers1989}. 
TEL extends topological modal logic by providing tools for reasoning about what an agent can justify and know, given the available evidence. Given a topological space $\langle X, \topo \rangle$, a hypothesis $P \subseteq X$ is \emph{justified} if $U \subseteq P$ for a dense open set $U$ (evidence for $P$ that is consistent with all available evidence), and \emph{known} at a state $x \in X$ if the evidence $U$ is truthful ($x \in U$). 
However, classical TEL remains `\emph{resource-blind}', assuming any open set is accessible regardless of cost.

We bridge this gap by introducing \emph{semiring-annotated topological spaces} (seats) as a representation of resource use within topological models of evidence. Seats extend topological spaces $\langle X, \topo\rangle$ with an annotation function $\ann_K\colon \topo \times X \to \powset (K)$ for a semiring of resources $K$; the intuition is that $\ann_K(U, x)$ is the collection of resources $a \in K$ that are sufficient to access evidence $U \in \topo$ at state $x \in X$. Intuitively, this allows us to reason not just about what is \emph{observable} in principle, but what is \emph{affordable} given a budget. 
We demonstrate that this framework is not only a natural extension of TEL, but also unifies structures from diverse areas of computer science, such as programming semantics, database security, robotics and distributed systems. 
Our main contribution is the development of \emph{seat-based epistemic logics}, with formulas $\F_a\f$ which intuitively express the availability of evidence for $\f$ given the resource $a$, as well as with formulas expressing TEL-style epistemic justification and knowledge in the resource-constrained setting. 

Our main technical results include strong completeness theorems for the minimal seat-based epistemic logic as well as for various extensions characterizing naturally defined classes of seats. That is, we provide sound and strongly complete axiomatizations of logics capable of reasoning about the trade-offs between the precision of knowledge and the cost of the underlying observations.
Furthermore, we define suitable notions of bisimulation and disjoint unions in our framework and use them to delineate the expressive capabilities of these logics.

Although several frameworks incorporate resource constraints into epistemic logic, typically by annotating Kripke or neighbourhood models with numerical costs, our approach offers two advantages. Firstly, by using arbitrary semirings instead of specific numerical scales, we obtain a general framework capable of modelling non-linear resource structures, such as security clearances or multi-dimensional budgets. Secondly, by situating our framework within TEL, we connect resource-aware epistemic logics and existing work on the topology of observable properties. A more detailed discussion of related work is provided in Section~\ref{sec:RelatedWork}.

The paper is structured as follows. In Section~\ref{sec:motivation}, we recall the TEL framework, present some examples of its use, and motivate the need for a representation of resources. In Section~\ref{sec:seats}, we introduce semiring-annotated topological spaces (seats), which provide a rigorous mathematical foundation for our intuitive considerations. Section~\ref{sec:Logics and completeness I} introduces seat-based epistemic logics and establishes soundness and strong completeness for logics based on various natural classes of seats. In Section~\ref{sec:Logics and completeness II}, we extend the logical language of Section~\ref{sec:Logics and completeness I} with the  global modality and show how the resulting framework gives rise to resource-aware generalizations of the core epistemic operators of \cite{BaltagEtAl2022}. We also establish soundness and strong completeness for our logics extended with the global modality. In Section~\ref{sec:Undefinability}, we prove undefinability results using appropriately defined notions of disjoint unions and bisimulations. Section~\ref{sec:RelatedWork} discusses related work in more detail. We conclude in Section~\ref{sec:Conclusion}, providing a summary and an outlook for future work. Full proofs and further details are provided in Appendix~\ref{app}.

\section{Motivation}\label{sec:motivation}

In this section, we outline the TEL framework (Section \ref{sec:motivation1}) and argue for extending it to include a representation of the resources used to obtain evidence (Section \ref{sec:motivation2}). For more details about TEL, we refer the reader to \cite{BaltagEtAl2022}.

\subsection{Topology of evidence}\label{sec:motivation1}
We assume that the reader is familiar with basic topological notions such as topological space, open set, interior, closure, density, basis, subbasis, etc.\ (see \cite{Munkres2000}, for example). 

Let $\langle X, \topo\rangle$ be a topological space where $\topo$ is generated by a basis $\base$ with a subbasis $\subbase$. It is suggested in \cite{BaltagEtAl2016,BaltagEtAl2022} to view elements of $\subbase$ as pieces of \emph{direct evidence} and elements of $\base$ as \emph{combined evidence}. Open sets $U \in \topo$ correspond to \emph{arguments}, or collections of evidence that can be used to support a conclusion. 
A proposition $P \subseteq X$ is \emph{supported} by $U \in \topo$ if $U \subseteq P$. Thus, the interior $\Int (P)$, is the weakest argument supporting $P$. Equivalently, $\Int (P)$ can be seen as the proposition stating that there is truthful evidence supporting $P$. 
On the evidential reading, an argument $U \in \topo$ is \emph{dense} in $\langle X, \topo\rangle$ if it is consistent with all non-empty (i.e.\ consistent) arguments $V \in \topo$; a dense $U$ cannot be contradicted by any consistent evidence. 
Baltag et al.\ \cite{BaltagEtAl2016,BaltagEtAl2022} link the topological notion of density to the notions of \emph{epistemic justification} and \emph{knowledge}. 
A hypothesis $P \subseteq X$ is \emph{justified} if there is a dense open $U \in \topo$ such that $U \subseteq P$ (i.e.\ $U$ supports $P$ and cannot be refuted by any consistent evidence); and $P$ is \emph{known} at $x \in X$ if there is a dense open neighborhood of $x$ that supports $P$, i.e.\ a dense $U \in \topo$ such that $U \subseteq P$ and $x \in U$. 

Recall that if $\topo$ is generated by a basis $\base$, then $U \in \topo$ is dense iff $B \cap U \neq \varnothing$ for all $B \in \base {\setminus} \{ \varnothing \}$. Equivalently, $U = \bigcup_{i \in I} B_i$ for $\{ B_i \}_{i \in I} \subseteq \base$ is dense iff for all $B \in \base {\setminus} \{ \varnothing \}$ there is $B_i$ such that $B \cap B_i \neq \varnothing$. Hence, $P$ is justified if there is a collection  $\{ B_i \}_{i \in I} \subseteq \base$ such that all $B_i$ support $P$, and for each `objection' $B \in \base {\setminus} \{ \varnothing \}$ there is a `response' $B_i$ consistent with $B$.

The following examples illustrate the TEL approach.

\begin{example}[Observing binary streams \cite{Smyth1992,Vickers1989}]\label{exam:streams}
    Consider a device that outputs a binary sequence, such as a server or a sensor. Countable (finite or infinite) words $w \in \{ 0, 1 \}^{\infty}$ represent the outputs the device would yield if given infinite time; finite words $w \in \{ 0, 1 \}^{*}$ represent finite observations of these outputs. Recall that $w \in \{ 0, 1 \}^{\infty}$ is a directed-complete partially ordered set under the prefix order $\sqsubseteq$, where $w \sqsubseteq u$ means that $w$ is a prefix of $u$. 
    For any set $O \subseteq \{ 0, 1 \}^{*}$ of possible observations such that $\epsilon \in O$, the collection of ${\uparrow}w = \{ u \in \{ 0, 1 \}^{\infty} \mid w \sqsubseteq u \}$ for $w \in O$ forms a basis for a topology $\topo_O$ on $\{ 0, 1 \}^{\infty}$, which is typically coarser than the Scott topology on $\{ 0, 1 \}^{\infty}$ (a special case where $O = \{ 0, 1 \}^{*}$). This is the topology of observable properties of binary words, given the set $O$ of possible observations. 
    An open set is dense if it is consistent with every possible observation. The set of possible observations is given by the context and may depend on the `actual computation' $w \in \{ 0, 1\}^{\infty}$. For instance, we can have $O(w) = \{ v \in \{ 0, 1 \}^{*} \mid v \sqsubseteq w \}$. In this case, $P \subseteq \{ 0, 1\}^{\infty}$ is known at $w \in \{ 0, 1\}^{\infty}$ if ${\uparrow} w \subseteq P$.
\end{example}

\begin{example}[Role-Based Access Control in databases \cite{SandhuEtAl1996,BonattiEtAl2002}] \label{exam:databases}
    Fix a set $\mathit{DB}$ of databases and a finite set $R$ of \emph{user roles}. Let $d\colon R \to \powset (\mathit{DB})$ be the permission assignment, mapping each role to its accessible databases. 
    Let $\Omega$ be the state space and $p\colon \powset (\mathit{DB}) \to \powset (\Omega)$ an anti-monotone function mapping a set of databases to the states consistent with their content. The composition $pd = p \circ d \colon R \to \powset (\Omega)$ represents the view of the system available to a specific role. 
    We model \emph{qualifications} as subsets of roles. Let $\mathit{QL} \subseteq \powset(R)$ be a collection of role sets closed under union and intersection. A qualification $a \in \mathit{QL}$ represents a requirement, for instance `the user must possess one of the roles in $a$'. 
    A proposition $P \subseteq \Omega$ is \emph{qualification-observable} if there exists a qualification $a \in \mathit{QL}$ such that $pd(r) \subseteq P$ for all $r \in a$. 
    Intuitively, if a user satisfies the qualification $a$ (possesses \emph{some} role in $a$), they can verify $P$ regardless of which specific role in $a$ they hold. This captures a notion of robust access. 
    The collection of qualification-observable propositions forms the basis of a topology $\topo_\mathit{QL}$ on $\Omega$. 
    Within this topology, a proposition $P$ is \emph{justified} if it is qualification-observable (since $\topo_{\mathit{QL}}$ is closed under supersets) and its complement cannot be observed by any non-empty qualification. This corresponds to a property that is verifiable by some group of users and cannot be refuted by any other group.
\end{example}

\begin{example}[Exploring a graph \cite{DudekEtAl1991}] \label{exam:graphs}
   Let $\langle V, E \rangle$ be a connected graph and $\Omega$ a state space, e.g.\ the set of all graphs on $V$. Assume that every $v \in V$ provides some information or `local perspective' on the `global state'. This may be represented by a function $f \colon V \to \powset (\Omega)$, where $f(v)$ is the set of global states consistent with the information available in $v$ (e.g.\ the number of its neighbors). 
   The function $f$ can be lifted to paths $t \in V^{*}$ over $\langle V, E \rangle$ by defining $f(\langle v_1, \ldots, v_n \rangle) = \bigcap_{i = 1}^{n} f(v_i)$. Intuitively, this corresponds to a robot traversing the path $t = \langle v_1, \ldots, v_n \rangle$, observing the vertices along the path and the information they provide. 
   The function $f$ can also be lifted to sets of paths $L \subseteq V^{*}$ by defining $f(L) = \bigcap_{t \in L} f(t)$. Intuitively, $f(L)$ is the information obtained by traversing all paths in $L$, e.g.\ by a group of robots in a parallel exploration of the graph. Note that $f(\varnothing) = \Omega$.  
   The set $\{ f(L) \mid L \text{ is a finite set of paths in } \langle V, E \rangle \}$ is the basis of a topology $\topo_f$ on $\Omega$. Intuitively, this is the topology of properties of the global state that can be observed locally by a group of robots exploring the graph. We may require that all `legal' paths start in some fixed `initial vertex' $v_0$. 
   Now assume that $\langle V, E \rangle$ is the `actual state'  and $v_0 \in V$ is the  initial vertex. A property $P \subseteq \Omega$ is known (or, better, `knowable') if $\langle V, E \rangle \in P$ (the actual state has the property), there is a collection $\{ L_i \}_{i \in I}$ of finite sets of paths $L_i$ on $\langle V, E \rangle$ starting in $v_0$ such that $f(L_i) \subseteq P$ (the property is verifiable by a number of finite explorations of the graph) and $f(L) \cap f(L_i) \neq \varnothing$ for each finite non-empty set of paths $L$ (the findings of each possible `counter-exploration' $L$ are consistent with the finding of some of the options in $\{ L_i \}_{i \in I}$). 
\end{example}

\begin{example}[Agents in distributed systems \cite{FaginEtAl1995}]\label{exam:agents}
    Let $A$ be a set of agents in a distributed system and  $\Omega$ be the set of all global states of the system (both sets may be infinite). We assume the standard partition model: for each agent $n \in A$, let $\sim_n$ be an equivalence relation on $\Omega$ where $s \sim_n s'$ indicates that agent $n$ cannot distinguish state $s$ from $s'$. 
    For a finite group of agents $G \subseteq A$, we define $[s]_G = \bigcap_{n \in G} [s]_n$, representing the \emph{distributed knowledge} of the group $G$ in $s$, i.e.\ the combined information of the group. 
    The collection of all $[s]_G$ for $s \in \Omega$ and finite $G \subseteq A$ forms a basis for a topology $\topo_A$ on $\Omega$. 
    This topology captures the properties of the system that are observable by an external observer $N$, who can query the distributed knowledge of any group. 
    A proposition $P \subseteq \Omega$ is justifiable if there is a collection $\{ [s_i]_{G_i} \}_{i \in I}$ such that $\bigcup_{i \in I} [s_i]_{G_i} \subseteq P$ ($N$ considers it possible that $P$ is distributed knowledge in groups $G_i$) such that for all $[t]_H \neq \varnothing$ there is $i \in I$ with $[s_i]_{G_i} \cap [t]_H \neq \varnothing$ -- no matter which group $H$ an adversary consults or which potential state $t$ they propose, their distributed knowledge is consistent with at least one piece of evidence for $P$. 
    The proposition $P$ is known at $s$ if the above holds and $s \in [s_i]_{G_i}$ for some $i \in I$.
\end{example}

\subsection{Resources}\label{sec:motivation2}
In practice, accessing evidence requires \emph{resources}. Real-life agents operate within resource \emph{budgets}, meaning that the amount of resources they can spend on obtaining evidence is limited. Consequently, a finer-grained, resource-aware notion of justification comes to the forefront.

\smallskip

\emph{Example~\ref{exam:streams}, continued.}
Observable properties of words are established by observing finite words $w \in O$; the resource spent is the \emph{time} needed to observe a given finite word. We may represent the time needed to observe $w \in O$ by its \emph{length} $|w| \in \mathbb{N}$. 
We say that time $n$ is sufficient to obtain evidence $U \in \topo_O$ if there is $w \in O$ such that ${\uparrow}w \subseteq U$ and $|w| \leq n$. We express this by writing $n \to U$. 
Note that `$\to$' has a number of general properties, for instance: (i) \emph{Resource strengthening:} if $n \to U$ and $m \in \mathbb{N}$, then $\max (n, m) \to U$ (`if a resource is sufficient for $U$, then any stronger resource is also sufficient for $U$'); and (ii) \emph{Evidence weakening:} if $n \to U$ and $U \subseteq V$ for an open set $V$, then $n \to V$ (`if a resource is sufficient for $U$, then it it sufficient for any weaker evidence'). 

\smallskip

\emph{Example \ref{exam:databases}, continued. }      
Evidence is obtained by fulfilling a \emph{qualification}, so we may think of $\mathit{QL}$ as the collection of resources. We say that $a \in \mathit{QL}$ is sufficient to obtain $U \in \topo_\mathit{QL}$ if $pd(r) \subseteq U$ for all $r \in a$. As before, we indicate this by $a \to U$ and observe that analogues of properties (i) and (ii) hold here : (i) if $a \to U$, then $a \cap b \to U$ for all $b \in \mathit{QL}$ (`if satisfying $a$ gives access to information that supports $U$, then satisfying ``$a$ and $b$'' does so as well'); and (ii) if $a \to U$ and $U \subseteq V$ for an open set $V$, then $a \to V$. 
Additional properties are discernible: (iii) \emph{Resource choice:} if $a \to U$ and $b \to U$, then $a \cup b \to U$ (`if satisfying $a$ and satisfying $b$ are both sufficient for $U$, then satisfying ``$a$ or $b$'' is sufficient'); and (iv) \emph{Resource combination:} if $a \to U$ and $b \to V$, then $a \cap b \to (U \cap V)$ (`if $a$ and $b$ are sufficient for $U$ and $V$, respectively, then their combination, ``$a$ and $b$'', is sufficient for the combined evidence $U \cap V$'). 
Note that analogues of (iii) and (iv) hold in Example \ref{exam:streams} as well if `$n$ or $m$' is represented by $\min (n,m)$ and `$n$ together with $m$' is represented by either $\max(n,m)$ or $n + m$ (the former is sufficient: if $v \in {\uparrow}w$ and $v \in {\uparrow}u$ for $|w| \leq n$ and $|u| \leq m$, then both $w$ and $u$ are prefixes of $v$).

\smallskip

\emph{Example \ref{exam:graphs}, continued. } 
Assume that $\langle V, E \rangle$ is a \emph{weighted} graph, say with positive rational edge weights. The weight $E(v,w)$ could represent physical distance, battery requirements, etc.\ 
Let $E(\langle v_1, \ldots, v_n \rangle) = \sum_{i = 1}^{n-1} E(v_i, v_{i+1})$ for a path and let $E(L) = \sum_{t \in L} E(t)$ for a finite set of paths $L$ (an `exploration'). 
We say that $q \in \mathbb{Q}_{>0}$ is sufficient to obtain the observable $U$ iff there is an exploration $L$ such that $f(L) \subseteq U$ and $E(L) \leq q$. We write $q \to U$ as before. We leave it to the reader to verify that the properties (i) -- (iv) identified in the first two examples also hold here, if `$q_1$ or $q_2$' is $\min(q_1, q_2)$ and `$q_1$ together with $q_2$' is $q_1 + q_2$. 
We add a fifth property: (v) $0 \to \Omega$, meaning that the tautologous open set $\Omega$ is available `for free' (note that if $L = \varnothing$, then $f(L) = \Omega$ and $E(L) = 0$; it is known that $\Omega$ holds without exploring the graph at all). Analogues of this property hold in the previous examples: $0 \to \{ 0, 1 \}^{\infty} = {\uparrow}\epsilon$ in Example \ref{exam:streams}, where $|\epsilon| = 0$ and $R \to \Omega$ in Example \ref{exam:databases}, since $pd(r) \subseteq \Omega$ for all $r \in R$.

\smallskip

\emph{Example \ref{exam:agents}, continued. } 
Resources used by the external observer $N$ to observe the system are \emph{groups of agents} $G \subseteq A$. We must distinguish between local and global sufficiency of a group $G$ for an open set $U$. We say that $G$ is sufficient for $U$ at state $s$, denoted by $G \xto{s} U$, if $[s]_G \subseteq U$; on the other hand, $G$ is sufficient for $U$ globally, $G \to U$, if $G \xto{s} U$ for some $s$.  
As before, the analogues of the general properties (i) -- (v) hold for $\xto{s}$, however some do not hold for $\to$. In particular, if $G \xto{s} U$ and $H \xto{s} V$, then $G \cup H \xto{s} (U \cap V)$, but the variant of this property with $\to$ instead of $\xto{s}$ fails since we can have $[s]_G \cap [t]_H = \varnothing$ and $\varnothing \neq [u]_D$ for all $u \in \Omega$ and $D \subseteq A$. Note that (v) holds since $\Omega = [s]_\varnothing$ for all $s \in \Omega$.

\smallskip

The following framework emerges. Firstly, in each example the collection of resources forms an algebra, namely a (zero-free) \emph{semiring} $\langle K, \oplus, \odot, \1 \rangle$. The operation $\odot$ represents resource \emph{combination} (resource $a\odot b$ can be read as `\emph{$a$ together with $b$}' or `\emph{$a$ and $b$}') and the operation $\oplus$ represents resource \emph{choice} (resource $a \oplus b$ can be read as `\emph{$a$ or $b$}'). As usual, we use $ab$ instead of $a \odot b$. The multiplicative unit $\1$ represents using \emph{no resources} (note that using the combination $a\1$ means using $a$). In Example~\ref{exam:streams}, we use $\langle \mathbb{N}, \min, \max, 0 \rangle$, in Example~\ref{exam:databases} we use the distributive lattice $\langle \mathit{QL}, \cup, \cap, R \rangle$, in Example~\ref{exam:graphs} we use $\langle \mathbb{Q}_{\geq 0}, \min, +, 0 \rangle$ and in Example~\ref{exam:agents} we use the unusual $\langle \powset (A), \cup, \cup, \varnothing \rangle$. 
The use of semirings to provide an abstract representation of resources is widespread in computer science, ranging from constraint satisfaction problems \cite{BistarelliEtAl1997} to automata theory \cite{KuichSalomaa1986} and databases \cite{GreenEtAl2007}.

Secondly, each example includes a topological space $\langle X, \topo \rangle$ where the open sets correspond to observable properties or propositions. Open sets $U \in \topo$ are \emph{annotated with subsets of the respective semirings $K$}; sometimes it is natural to parametrize the annotation by a state $x \in X$. Intuitively, $a \xto{x} U$ (meaning that $a \in K$ belongs to the annotation of $U \in \topo$ relative to $x \in X$) expresses the idea that resource $a$ is sufficient to obtain evidence $U$ at state $x$. The annotations satisfy a number of properties (for all $a \in K$, $U, V \in \topo$ and $x \in X$):
\begin{itemize}
    \item[(i)] \emph{Resource strengthening:} if $a \xto{x} U$, then $ab \xto{x} U$ and $ba \xto{x} U$ for all $b \in K$ (if $a$ is sufficient for $U$, then `$a$ together with $b$' is sufficient, in any order); 
    \item[(ii)] \emph{Evidence weakening:} if $a \xto{x} U$ and $U \subseteq V$, then $a \xto{x} V$ (if $a$ is sufficient for $U$, then $a$ is sufficient for any weaker $V$);
    \item[(iii)] \emph{Resource choice:} if $a \xto{x} U$ and $b \xto{x} U$, then $a \oplus b \xto{x} U$ (if both $a$ and $b$ are sufficient for $U$, then whichever is chosen will be sufficient);
    \item[(iv)] \emph{Resource combination:} if $a \xto{x} U$ and $b \xto{x} V$, then $ab \xto{x} (U \cap V)$ and $ba \xto{x} (U \cap V)$ (if $a$ is sufficient for $U$ and $b$ is sufficient for $V$, then `$a$ together with $b$' is sufficient for $U \cap V$);
    \item[(v)] \emph{Available tautologies:} $a \xto{x} X$ for some $a \in K$ (tautologous evidence can be obtained).
\end{itemize}

Let $\ann (U, x) = \{ a \mid a \xto{x} U \}$. Note that conditions (i) and (iii) imply that $\ann (U, x)$ is an \emph{ideal} of $K$ (possibly empty or improper). Conditions (iv) and (v) imply that, for all $x \in X$, the set $\{ U \mid \ann (U, x) \neq \varnothing \}$ is a basis for a topology $\topo(x)$ on $X$, which is coarser than $\topo$.

Our examples satisfy a stronger variant of condition (v), namely $\1 \xto{x} X$, which states that `tautologous evidence is free'. However, we do not generally assume this stronger form, since it is easy to imagine situations where substantial resources need to be spent on obtaining tautologous evidence (consider hard mathematical proofs, for example).

It is usually assumed that semirings contain an additive identity that is also a multiplicative annihilator, that is, an element $\0 \in K$ such that $a \oplus \0 = a$ and $a\0 = \0a = \0$ for all $a \in K$. Intuitively speaking, $\0$ represents the strongest resource (using `$a$ or $\0$' is equivalent to using $a$ since $a$ is contained in $\0$; using `$a$ together with $\0$' is equivalent to using $\0$). We note that in the presence of condition (i) above, condition (v) can be equivalently formulated as
\begin{itemize}
    \item[(v')] $\0 \xto{x} X$. 
\end{itemize}
This turns out to be a more practical formulation of the underlying assumption, hence we use the standard definition of semiring in computer science (i.e.\ with zero) in a manner that is consistent with our resource-based interpretation.

Turning to justification, assume that an agent has a limited resource budget. The agent may only be able to observe words up to a certain length (Example \ref{exam:streams}), satisfy a specific set of qualifications (Example \ref{exam:databases}), have only limited resources for exploring a graph (Example \ref{exam:graphs}) and consult specific groups of agents (Example \ref{exam:agents}). \emph{What can the agent justify and know?} 
It is reasonable to assume that if the agent can justify a proposition $P \subseteq X$, then there is some $P$-supporting evidence $U \in \topo$ that \emph{can be obtained within the agent's budget}. However, as in TEL, we should also require the agent to be able to \emph{defend their evidence against relevant counterexamples}. The set of these counterexamples itself is typically limited by a resource budget, in general belonging to an `opponent' (e.g.\ other agents, `nature' or the justifying agent itself).

We arrive at the notions of \emph{justification and knowledge on a budget}. These notions are relevant in situations where evidence must be both proposed and opposed, but both processes are limited by resource restrictions. 

In Example~\ref{exam:streams}, one can imagine the set of relevant observations $O$ being generated by a group of `verifiers' that need to be mutually consistent. In Example~\ref{exam:databases}, one may be interested in what information can be justified by users satisfying specific qualifications (e.g.\ to prevent them from observing certain security-sensitive facts). In Example~\ref{exam:graphs}, justification on a budget corresponds to asking which hypotheses about the structure of the graph are `stable' under specific limitations on explorations (e.g.\ time or battery consumption). In Example~\ref{exam:agents}, budget restrictions pertain to limits of group communication, as only certain groups of agents are assumed to be able to provide evidence (i.e.\ to pool their information to arrive at distributed knowledge).

\section{Semiring-annotated topological spaces}\label{sec:seats}

In this section, we introduce semiring-annotated topological spaces formally (Section~\ref{sec:seats1}), look at some special classes of these structures (Section~\ref{sec:seats2}), and investigate epistemic propositional operators based on them (Section~\ref{sec:seats3}). These operators will play a crucial role in the semantics of the epistemic logics introduced in Sections~\ref{sec:Logics and completeness I} and~\ref{sec:Logics and completeness II}.

\subsection{Seats}\label{sec:seats1}

We now introduce structures in which open sets of a topological space are annotated by elements of a semiring. While all examples in Section~\ref{sec:motivation} use idempotent and commutative semirings, there is no mathematical need to limit our framework to this specific class.

\begin{definition}[Seats]\label{def:seat}
A \emph{semiring-annotated topological space} (\emph{seat}) is a tuple $\bm{S}=\langle X, \topo, \ann_K\rangle$ where $\langle X, \topo\rangle$ is a topological space, $K= \langle K, \oplus, \odot, \0,\1 \rangle$ is a semiring and \[ \ann_K \colon \topo \times X \to \powset (K) \, \] satisfies the following (for all $a \in K$, $U, V \in \topo$, and $x \in X$):
\begin{eqnarray}
& a \in \ann_K(U,x) \!\And\! b \in K \,\Rightarrow\, ab, ba \in \ann_K(U,x) \label{A(i)}\\
& U \subseteq V \,\Rightarrow\, \ann_K(U,x) \subseteq \ann_K(V, x) \label{A(ii)}\\
& a, b \in \ann_K(U, x) \,\Rightarrow\, a \oplus b \in \ann_K (U,x) \label{A(iii)}\\
& a \in \ann_K(U,x) \!\!\And\!\! b \in \ann_K(V, x) \Rightarrow ab \in \ann_K(U \cap V, x) \label{A(iv)}\\
& \0 \in \ann_K(X, x) \label{A(v)}
\end{eqnarray}
We define $\evid_a(x) = \{ U \in \topo \mid a \in \ann_K(U, x) \}$ and $\base(x) = \{ U \in \topo \mid \ann_K(U,x) \neq \varnothing \}$.
\end{definition}

Intuitively, $\ann_K(U, x)$ is the set of resources that are sufficient for obtaining evidence $U$ at state $x$. As noted above, $\ann_K(U,x)$ is an ideal of $K$ (possibly empty or improper). If $\ann_K(U, x) = \varnothing$, then $U$ cannot be obtained in $x$. We will write $\ann_x(U)$ instead of $\ann_K(U, x)$ when $K$ is clear from the context or immaterial. 
The set $\evid_a(x)$ comprises pieces of evidence that can be obtained using resource $a$ at state $x$ and $\base(x)$ contains evidence that can be obtained using \emph{some} resource in $x$. As noted above, $\base(x)$ is a basis for a topology on $X$ which we will denote by $\topo(x)$. 
Note that $\base(x) = \evid_\0(x)$.

\begin{remark}
The term `semiring-annotated topological space' is an analogy to semiring-annotated databases \cite{GreenEtAl2007} and the interpretation of $\oplus$ in terms of resource choice is directly inspired by the corresponding usage therein. 
\end{remark}

We now specify the seat used to formalize Example~\ref{exam:streams}; seats related to the remaining examples of Section~\ref{sec:motivation} are discussed later on (Examples~\ref{exam:3.6}--\ref{exam:3.8}).

\begin{example}
    Let $\langle \{ 0, 1 \}^{\infty}, \topo \rangle$ be $\{ 0, 1 \}^{\infty}$ with the Scott topology. We assume that for each $w \in \{ 0, 1 \}^{\infty}$ we have a set $O(w) \subseteq \{ 0, 1 \}^{*}$ of possible observations of $w$ such that $\{ u \sqsubseteq w \mid u \in \{ 0, 1 \}^{*} \} \subseteq O(w) \subseteq \{ 0, 1 \}^{*}$. That is, $O(w)$ contains all `correct' finite observations but may also include additional observations representing `noise'. Let $K = \langle \mathbb{N} \cup \{ \infty\}, \min, \max, \infty, 0 \rangle$ and define \[ \ann_K (U, w) = \{ a \in K \mid \exists u\in O(w) \colon  {\uparrow}u \subseteq U \!\And\! |u| \leq a \}. \] 
    It is easily verified that this is a seat.
\end{example}

\begin{remark}[Strength preorders]\label{rem:strength}
`Resource strengthening' can be seen as monotonicity with respect to a `strength preorder', though in arbitrary semirings this can have different interpretations. For one, the \emph{right} (resp.\ \emph{left}) \emph{multiplicative preorder}, where $a \le_R b$ iff $\exists c\colon b = ac$ (resp.\ $a \le_L b$ iff $\exists c\colon b = ca$), meaning `$b$ can be obtained as $a$ together with some resource $c$'. Monotonicity under this preorder corresponds to property \eqref{A(i)}.

Alternatively, we may consider the \emph{additive preorder} $a \sqsubseteq b$ iff $\exists c\colon b = a \oplus c$. This is the `standard' preorder on semirings (which is a partial order if $\oplus$ is idempotent). We can take $\sqsubseteq^{-1}$ as a strength preorder (note that $a \sqsubseteq^{-1} \0$ for all $a$) where $b \sqsubseteq^{-1} a$ read as `using $a$ is equivalent to using $b$ or some other resource $c$'. If $\oplus$ is idempotent, this takes the (more natural) meaning $b \sqsubseteq^{-1} a$ iff $a \oplus b = b$ (`using $a$ or $b$ is equivalent to using $b$'). The latter implies that $b$ is `contained in' $a$ (since even by using $a$ one uses $b$ as well). 
Since $\sqsubseteq^{-1}$ is natural only in a restricted class of semirings (idempotent ones) and the associated monotonicity condition defines a restricted kind of semiring ideal (strong ideal), we do not assume $\sqsubseteq^{-1}$-monotonicity as a basic condition. 
However, in many semirings the multiplicative strength preorder coincides with the additive one, e.g.\ in Examples~\ref{exam:streams}--\ref{exam:graphs}.
\end{remark}

\subsection{Special seats}\label{sec:seats2}

Some of the following seat properties are reasonable in specific contexts. 

\begin{definition}\label{def:SpecialSeats}
A seat $\langle X, \topo, \ann_K\rangle$ is 
\begin{itemize}
\item \emph{strong} if $\ann_x(U)$ is a strong ideal for all $U \in \topo$, that is, $a \oplus b \in \ann_x(U) \,\Rightarrow\, a,b \in \ann_x(U)$;
\item \emph{$\1$-bounded} if $\1 \in \ann_x(X)$ for all $x \in X$ (consequently, $\ann_x(X) = K$ for all $x \in X$);
\item \emph{$\0$-bounded} if $\0 \in \ann_x(\varnothing)$ for all $x \in X$ (consequently, $\0 \in \ann_x(U)$ for all $U \in \topo$);
\item \emph{uniform} if $\ann_x(U) = \ann_y(U)$ for all $x,y \in X, U\in \topo$.
\end{itemize}
A seat is \emph{bounded} if it is both $\0$-bounded and $\1$-bounded.
\end{definition}
   
Strong seats represent the idea that $\ann_K$ is monotonic under the additive strength preorder $\sqsubseteq^{-1}$ (see Remark \ref{rem:strength}).
In $\1$-bounded seats, tautologous evidence $X$ can always be obtained using `no resource' $\1$ (i.e.\ $X$ is `for free') and so $\ann_x(X)$ is the improper ideal $K$. In $\0$-bounded seats, contradictory evidence $\varnothing$ can be accessed using the strongest resource $\0$; in many settings, this reflects the idea that $\0$ represents an inaccessible (`infinite') resource. Using \eqref{A(ii)}, in $\0$-bounded seats every piece of evidence can be obtained using \emph{some} resource and, consequently, $\topo(x) = \topo$ for all $x$. 
In uniform seats, the resources needed to obtain evidence do not depend on the state. Alternatively, uniform seats can be construed as containing an `actual state' $x_0 \in X$ and parametrizing the annotation function \emph{according to that state}. For uniform seats, we will write $\ann(U)$ instead of $\ann_x(U)$.

\begin{example}\label{exam:3.6}
    An example of a strong, uniform and bounded seat 
    is derived from Example~\ref{exam:databases}, where $\mathit{QL} \subseteq \mathcal{P}(R)$ forms a bounded distributive lattice, the topological space is $\langle \Omega, \topo_{\mathit{QL}} \rangle$ and $\ann(U) = \{ a \in \mathit{QL} \mid \forall r\in a \colon pd(r) \subseteq U \}$. 
\end{example}

\begin{example}\label{exam:3.7}
    To obtain a seat from Example \ref{exam:graphs}, fix a non-weighted graph $\langle V, D \rangle$, a weight function $E \colon D \to \mathbb{Q}_{\geq 0}$ and a starting point $v \in V$. Let $\Omega$ be the set of pairs $\langle G(V), u \rangle$, where $G(V)$ is a non-weighted graph on $V$ and $u \in V$ is a designated `starting state' of $G(V)$. The semiring of weights is $K = \langle \mathbb{Q}_{\geq 0}^{\infty}, \min, +, \infty, 0 \rangle$. 
    Recall from Example \ref{exam:graphs} and its continuation in Section \ref{sec:motivation2} that, given a local information function $f\colon V \to \powset (\Omega)$, the topology $\topo_f$ on $\Omega$ is generated by the set of $f(L)$ for $L$ a finite set of paths on $\langle V, D \rangle$. We define $\ann_x (U)$ as comprising $a \in K$ such that there is $L$ that starts in $v$ such that $f(L) \subseteq U$ and $E(L) \leq a$, where $v$ is the starting state of the graph $x \in \Omega$. This seat is strong and bounded, but it is not uniform.
\end{example}

\begin{example}\label{exam:3.8}
    The seat in Example \ref{exam:agents} (its exact formulation is left to the reader) is neither uniform nor strong. The latter holds by our choice of semiring: $U$ can be distributed knowledge in group $G \cup H$ without being distributed knowledge in $G$ or $H$. The seat is also not $\0$-bounded, since $\0 = \varnothing$. 
\end{example}

It is natural to define the \emph{cost} of a piece of evidence as the weakest resource that provides that piece of evidence. However, such  a  resource may not always exist.

\begin{definition}\label{def:CostSeat}
A seat $\langle X, \topo, \ann_K\rangle$ is a \emph{cost seat} if $K$ is idempotent and complete\footnote{That is, $\bigsqcup S$ exists for all $S \subseteq K$ and $\odot$ distributes over $\bigsqcup$ from both sides. In idempotent semirings, we write $\sqcup$ instead of $\oplus$.} and $\bigsqcup \ann_x(U) \in \ann_x(U)$, for all $U \in \topo$ and $x \in X$. 
The \emph{cost of $U$ in $x$} is then given by $\bigsqcup \ann_x(U)$. 
\end{definition}

Note that in a cost seat $\ann_x(U) \neq \varnothing$ since $\bigsqcup \varnothing = \0$. Given the intuitive reading of $\sqsubseteq^{-1}$ as an additive strength preorder, $\bigsqcup \ann_x(U)$ can be seen as the `weakest resource' that is contained in all $a \in \ann_x(U)$.

\begin{example}
    The seat corresponding to Example \ref{exam:streams} is not a cost seat since we can have $\ann_w(U) = \varnothing$. However, it can be turned into a cost seat by closing each $\ann_w(U)$ under suprema (the semiring of costs is complete). Given $O(w)$, the cost of $U$ in $w$ is the length of the shortest $u \in O(w)$ such that ${\uparrow}u \subseteq U$. This is $\infty$ if there is no such $u$. 
    The seat corresponding to Example \ref{exam:databases} is a cost seat if $K$ is finite; the cost of $U$ is the join of $\ann(U)$, that is, the weakest qualification that allows to observe $U$. 
    The seat corresponding to Example \ref{exam:graphs} is not a cost seat since the semiring of weights $\mathbb{Q}^{\infty}_{\geq 0}$ is not complete. 
    Finally, the seat in Example \ref{exam:agents} is a cost seat, but a peculiar one: the cost of $U$ in $s$ is the set of all agents.
\end{example}

\begin{example}[Cost seats from Borel measures]\label{exam:BorelCost}
Take a structure $ \langle X, \topo, \Sigma, \{ \mu_x \}_{x \in X} \rangle $ where $\langle X, \topo \rangle$ is a topological space, $\langle X, \Sigma \rangle$ is a measurable space such that $\topo \subseteq \Sigma$, and $\mu_x \colon \Sigma \to \mathbb{R}_{\geq 0}$ is a measure for all $x \in X$. Let $K$ be the complete tropical semiring of extended non-negative real numbers $\langle \mathbb{R}_{\geq 0} \cup \{ \infty \}, \inf, +, \infty, 0 \rangle$. Define $\ann_K (U, x) = \{ a \in \mathbb{R}_{\geq 0} \cup \{ \infty \} \mid \mu_x(X {\setminus} U) \leq a \}$. Then $\langle X, \topo, \ann_K \rangle$ is a cost seat where $\mu_x(X {\setminus} U)$ is the cost of $U$ in $x$. Intuitively, the cost of $U \in \topo$ in $x$ is the $x$-relative measure of the complement of $U$ (how easy it is to `miss' $U$). As a particular example, one may consider $\mu_x$ to be probability measures  -- for instance, the probability of a transition from $x$ ending up in a particular $P \in \Sigma$. In this example, cost is the amount of `risk' one is able to tolerate. This example provides a link to \emph{Markovian logics} \cite{HeifetzMongin2001,Zhou2007,Zhou2014,KozenEtAl2013}, discussed in Section \ref{sec:RelatedWork}. 
\end{example}

It can also be shown that cost seats arise from continuous functions between topological spaces: if $f \colon X \to X'$ is continuous, then $f^{-1}(U)$ is a cost of $U$ in the sense of our definition (the collection of open sets in $X$ is a frame and hence a complete semiring). However, the intuitive interpretation of arbitrary topologies as `costs' is more problematic. (In the context of dynamical topological logic \cite{KremerMints2005}, `next $U$' can possibly be interpreted  as the cost of $U$.)

\subsection{Epistemic operators}\label{sec:seats3}

In this section, we define several epistemic propositional operators on seats. These will be used in seat-based epistemic logics, which are studied in the following sections.  

\begin{definition}[The $\For$ operator]\label{def:For}
For a seat $\langle X, \topo, \ann_K\rangle$ and $a \in K$, we define $\For_a \colon \powset (X) \to \powset (X)$ as follows: 
$$
\For_a(P) = \{ x \mid \exists U \in \topo\colon U \subseteq P \And U \in \evid_a (x) \}.
$$
\end{definition}

Intuitively, $\For_a(P)$ is the set of states in which some evidence supporting $P$ is accessible using resource $a$ (`one can obtain $P$ \emph{for} $a$').

\begin{definition}[Weighted interior]
For a seat $\langle X, \topo, \ann_K\rangle$ and $a \in K$, we define $\Int_a \colon \powset (X) \to \powset (X)$ as follows: 
\[ \Int_a(P) = \{ x \mid \exists U \in \topo\colon x \in U \subseteq P \And U \in \evid_a(x)\}.\]
\end{definition}

$\Int_a(P)$ is the set of states in which some \emph{factive} evidence supporting $P$ is accessible using $a$ (`one can verify $P$ for $a$'). 

\begin{lemma}\label{lem:SoundnessHelp}
The following hold in all seats $\langle X, \topo, \ann_K\rangle$:
\begin{enumerate}
\item $P \subseteq Q \,\Rightarrow\, \For_a(P) \subseteq \For_a(Q)$;
\item $\Int_a (P) = \For_a (P) \cap \Int(P)$;
\item $\For_a \Int(P) = \For_a (P)$;
\item $\For_a(P) \cap \For_b(Q) \subseteq \For_{ab}(P \cap Q)$.
\end{enumerate}
\end{lemma}
\begin{proof}
1.\ Follows immediately from Definition~\ref{def:For}.

2. If $x \in \Int_a (P)$, then there is $U \in \topo$ such that $x \in U \subseteq P$ and $a \in \ann_x (U)$. Then clearly $x \in \For_a(P)$ and $x \in \Int(P)$. Conversely, if $x \in \For_a (P) \cap \Int (P)$, then there is $U \in \topo$ such that $U \subseteq P$ and $a \in \ann_x(U)$ and there is $V \in \topo$ such that $x \in V \subseteq P$. Then $a \in \ann_x(U \cup V)$ by \eqref{A(ii)} and $x \in U \cup V$; it follows that $x \in \Int_a(P)$.

3. Follows from the fact that $U \subseteq \Int (P)$ iff $U \subseteq P$, for all $U \in \topo$ and $P \subseteq X$. 

4. If $x \in \For_a(P) \cap \For_b(Q)$, then there are $U \in \evid_a(x)$ and $V \in \evid_b(x)$ with $U \cap V \subseteq P \cap Q$. By \eqref{A(iv)}, $U \cap V \in \evid_{ab}(x)$.
\end{proof}

It follows from the previous lemma that $\Int_a(P) \subseteq P$ and $\Int_a(P) \subseteq \Int_a(Q)$ if $P \subseteq Q$. However, $\Int_a$ is not necessarily a Kuratowski interior operator. The properties that fail are $\Int_a(X) = X$, $\Int_a(P) \cap \Int_a(Q) \subseteq \Int_a(P \cap Q)$ and $Int_a(P) \subseteq \Int_a(\Int_a(P))$; details are discussed in Appendix \ref{app:Inta}.

As discussed in Section~\ref{sec:motivation}, density is a crucial concept in TEL since it formalizes the notion of epistemic justification. In our resource-sensitive setting, these notions have compelling generalizations.

\begin{definition}
Let $\langle X, \topo, \ann_K\rangle$ be a seat. Given $a \in K$ and $x \in X$, we say a subset $S \subseteq X$ is \emph{$a$-dense in $x$} iff
\[ \forall U \in \evid_a(x) \colon U \neq \varnothing \implies S \cap U \neq \varnothing.\]
\end{definition}

Intuitively, $S$ is $a$-dense in $x$ iff it is consistent with all consistent pieces of evidence $U \in \topo$ that can be obtained using resource $a$ in $x$ (all $U$ within budget $a$). Note that $S$ is $\0$-dense in $x$ iff it is dense in the topology $\topo(x)$. 
In $\0$-bounded seats, $S$ is dense in $\topo$ in the standard sense iff $S$ is dense in $\topo(x)$ for arbitrary $x$ iff $S$ is $\0$-dense in $x$ for arbitrary $x$. 

\begin{definition}\label{def:BelAndKn}
For a seat $\langle X, \topo, \ann_K\rangle$ and any $a,b \in K$, we define $\Belop^a_b, \Knop^a_b \colon \powset (X) \to \powset (X)$ as follows:
\begin{align*}
x \in \Belop^a_b(P) &\text{ iff } \exists U \in \evid_a(x)\colon U \subseteq P {\And} U \text{ is $b$-dense in } x\\
x \in \Knop^a_b (P) &\text{ iff } \exists U \in \evid_a(x)\colon x \in U \subseteq P {\And} U \text{ is $b$-dense in } x
\end{align*}
\end{definition}

Intuitively, $x \in \Belop^a_b(P)$ if there is evidence $U$ supporting $P$ that can be obtained with budget $a$ in state $x$ and that is consistent with all consistent evidence that can be obtained with budget $b$. Thus, $P$ is justifiable given a budget $a$ for {\em finding} supporting evidence and a budget $b$ for finding {\em counterarguments} against it. Likewise, $x \in \Knop^a_b(P)$ if $P$ can be known given a budget $a$ for finding {\em truthful} supporting evidence and a budget $b$ for finding counterarguments against it.

In $\0$-bounded seats, these operators generalize the belief and knowledge operators of TEL \cite{BaltagEtAl2022}, which correspond to the cases $\Belop^\0_\0$ and $\Knop^\0_\0$, respectively. However, our setting is much more general and able to express more fine-grained notions of belief and knowledge.

\begin{example}
    For a `small' $\epsilon \in K$, $x \in \Belop^\epsilon_\epsilon (P)$ means that at $x$ there is `cheap' evidence for $P$ that is consistent with all (consistent) `cheap' evidence. This expresses the notion of \emph{superficial justification} by agents that are in a `rush' or whose justificatory budget is limited for another reason (e.g.\ think of the spread of unsubstantiated claims on social networks). 
    Note that $x \in \Belop^\epsilon_\epsilon (P) \cap X {\setminus} \Int(P)$ means that $P$ can be justified superficially but there is no \emph{factive} evidence to support it. This suggests a relation to the notion of \emph{disinformation}; $P$ can then be seen as a piece of misleading (false) information that some agents tend to accept without hesitation.

    One can interpret $x \in \Belop^\epsilon_b (P)$ as saying that there is `cheap' evidence for $P$ which however survives attacks by `more expensive' evidence. This situation arises when an agent is \emph{biased} towards $P$ (the agent does not have to exert too much effort to access evidence that supports $P$) but $P$ can be justified nevertheless.
\end{example}

\section{Logics and strong completeness}\label{sec:Logics and completeness I}

In this section, we introduce epistemic logics based on various classes of seats (we assume that the reader is familiar with standard modal logic \cite{BlackburnEtAl2001}). Our logics extend the modal logic $\mathbf{S4}$ with its topological semantics \cite{AielloEtAl2003,McKinseyTarski1944}, by adding operators that express the availability of evidence given a certain resource. The main results are strong completeness for the logic of all $K$-seats $\mathbf{S4}_K$ (Theorem~\ref{thm:CompletS4K}), the logic of all strong bounded seats $\mathbf{S4sb}_K$ (Theorem~\ref{thm:CompletS4sbK}) and the logic of all strong uniform bounded seats $\mathbf{S4sub}_K$ (Theorem~\ref{thm:CompletS4uK}).  

\begin{definition}
Let $\Prop$ be a countably infinite set of propositional variables. For a countable semiring $K$, we define the set of formulas of the language $\mathfrak{L}_K$ using the grammar
\[\f \Coloneqq p \mid \neg \f \mid \f \land \f \mid \F_a \f \mid \Box \f\]
where $p \in \Prop$ and $a \in K$. We define $\Box_a \f \coloneq \F_a \f \land \Box \f$. Other propositional operators and $\Diamond$ are defined as usual, furthermore we define $\Diamond_a \f \coloneq \neg \Box_a \neg \f$.
\end{definition}

Intuitively, the formula $\Box \f$ expresses that `there is truthful evidence supporting $\f$' and $\F_a \f$ expresses that `there is evidence for $\f$ which can be obtained using $a$'. Accordingly, the defined formula $\Box_a \f$ can be read as `there is truthful evidence for $\f$ which can be obtained using $a$'.    

\begin{definition}
A \emph{$K$-model} is a tuple $\bm{M} = \langle X, \topo, \ann_K, \val \rangle$ where $\langle X, \topo, \ann_K\rangle$ is a $K$-annotated topological space ($K$-seat) and $\val \colon \Prop \to \powset (X)$ is a valuation function. The satisfaction relation $\models$ is defined as follows:
\begin{itemize}
\item $\bm{M}, x \models p$ iff $x \in \val(p)$
\item $\bm{M}, x \models \neg \f$ iff $\bm{M}, x \not\models \f$
\item $\bm{M}, x \models \f \land \ff$ iff $\bm{M}, x \models \f$ and $\bm{M}, x \models \ff$
\item $\bm{M}, x \models \Box \f$ iff $\exists U\colon x \in U \in \topo \And U \subseteq \llbracket \f\rrbracket_{\bm{M}}$
\item $\bm{M}, x \models \F_a \f$ iff $\exists U\colon U \in \evid_a(x) \And U \subseteq \llbracket \f\rrbracket_{\bm{M}}$
\end{itemize}
where $\llbracket \f\rrbracket_{\bm{M}} = \{ x \mid \bm{M}, x \models \f \}$. Validity in $K$-models and $K$-seats is defined as expected. For $\mathsf{C}$ a class of $K$-seats we define the (local) semantic consequence relation $\models_{\mathsf{C}}$ as usual: 
$\Delta \models_\mathsf{C} \f$ iff $\bm{M},x \models \Delta \,\Rightarrow\, \bm{M},x \models \f$ for every state $x$ of every $K$-model $\bm{M}$ based on a seat in $\mathsf{C}$ (here, $\Delta \cup \{\f\}\subseteq \mathfrak{L}_K$). If $\mathsf{C}$ is the class of \emph{all} $K$-seats, we simply write $\Delta\models\f$.  
\end{definition}

Note that $\llbracket \Box \f\rrbracket_{\bm{M}} = \Int \llbracket \f\rrbracket_{\bm{M}}$ and $\llbracket \F_a \f\rrbracket_{\bm{M}} = \For_a \llbracket \f\rrbracket_{\bm{M}}$. By Lemma \ref{lem:SoundnessHelp},  it follows that $\llbracket \Box_a \f\rrbracket_{\bm{M}} = \Int_a \llbracket \f\rrbracket_{\bm{M}}$. 

\begin{definition}
Let $\mathbf{S4}_K$ be the extension of $\mathbf{S4}$ with 
\begin{gather}
    \F_a \f \to (\F_{ab} \f \land \F_{ba} \f)\label{a:Fmult}\\
    \F_a \f \land \F_b \ff \to \F_{a\oplus b}(\Box \f \lor \Box \ff)\label{a:Fplus}\\
	\F_a \f \land \F_b \ff \to \F_{ab}(\f \land \ff)\label{a:Fab}\\
    \F_\0 \top \label{a:Top}\\
	\F_a \f \to \F_a \Box \f \label{a:FBox}\\
	\dfrac{\f \to \ff}{\F_a \f \to \F_a \ff}\label{a:Fmono}
\end{gather}
The derivability relation $\proves_{\mathbf{S4}_K}$ is defined as usual. 
\end{definition}

The following lemma shows that $\mathbf{S4}_K$ is \emph{sound} with respect to the class of all $K$-seats (the proof is found in Appendix~\ref{app:SoudnessS4K}).

\begin{lemma}[Soundness $\mathbf{S4}_K$]\label{lem:SoudnessS4K}
Let $\Delta \cup\{ \f \} \subseteq \mathfrak{L}_K$. Then $\Delta\proves_{\mathbf{S4}_K}\f\,\Rightarrow\,\Delta \models \f$. 
\end{lemma}

To prove (strong) \emph{completeness} of $\mathbf{S4}_K$ with respect to the class of all $K$-seats, we use the \emph{canonical model} technique.

\begin{definition}\label{def:CanModelS4K}
Let  
$
\bm{M}^{\mathbf{S4}_K} = \langle X, \topo, \ann_K, \val\rangle
$,
where
\begin{itemize}
\item $X$ is the set of all maximal $\mathbf{S4}_K$-consistent theories $\Gamma \subseteq \mathfrak{L}_K$ ($| \f |$ is the set of $\Gamma \in X$ such that $\f \in \Gamma$);
\item $\topo$ is the topology on $X$ generated by the basis comprising $| \Box \f |$ for $\f \in \mathfrak{L}_K$ (see \cite[Lemma 3.2]{AielloEtAl2003});
\item $\ann_K(U, \Gamma) = \{a \in K \mid \exists \f\colon \F_a \f \in \Gamma \And |\Box \f| \subseteq U\}$;
\item $\val (p) = |p|$.
\end{itemize}
Note that $U \in \evid_a(\Gamma)$ iff $\exists \f\colon \F_a \f \in \Gamma \!\And\! |\Box \f| \subseteq U$. 
\end{definition}

We show that $\bm{M}^{\mathbf{S4}_K}$ is indeed the canonical $\mathbf{S4}_K$-model. 

\begin{lemma}\label{lem:CanModelS4K}
$\bm{M}^{\mathbf{S4}_K}$ is a $K$-model.
\end{lemma}
\begin{proof}[Proof sketch]
Straightforward application of the axioms. For instance, \eqref{A(iv)} is established as follows. If $a \in \ann_\Gamma(U)$ and $b \in \ann_\Gamma(V)$, then $\exists\f, \ff$ such that $\F_a \f \land \F_b \ff \in \Gamma$, $| \Box \f| \subseteq U$ and $| \Box \ff| \subseteq V$. Hence, $|\Box (\f \land \ff)| = |\Box \f| \cap |\Box \ff| \subseteq U \cap V$. From \eqref{a:Fab} we obtain $\F_{ab}(\f \land \ff) \in \Gamma$, which shows $ab \in \ann_\Gamma(U \cap V)$. For the rest, see Appendix~\ref{app:CanModelS4K}.  
\end{proof}

\begin{lemma}[Truth Lemma $\mathbf{S4}_K$]\label{lem:TruthLemmaS4K}
In the canonical $\mathbf{S4}_K$-model, $|\chi| = \llbracket \chi\rrbracket_{\bm{M}^{\mathbf{S4}_K}}$ for every $\chi \in \mathfrak{L}_K$.
\end{lemma}
\begin{proof}[Proof sketch]
Structural induction on $\chi$. We prove the case $\chi = \F_a \f$ here; the rest is deferred to Appendix~\ref{app:TruthLemmaS4K}.
If $\F_a \f \in \Gamma$, then $|\Box \f| \in \evid_a(\Gamma)$ by definition and $|\Box \f| \subseteq |\f|$ by $\mathbf{S4}$. The induction hypothesis yields $\bm{M}^{\mathbf{S4}_K}, \Gamma \models \F_a \f$. 
Conversely, suppose $\bm{M}^{\mathbf{S4}_K}, \Gamma \models \F_a \f$. By the induction hypothesis $\exists U \in \topo$ such that $U \in \evid_a(\Gamma)$ and $U \subseteq |\f|$. By definition of $\evid_a(\Gamma)$, $\exists \ff$ such that $\F_a \ff \in \Gamma$ and $|\Box \ff| \subseteq U$. It follows that $|\Box \ff| \subseteq | \f |$, which means that $\Box \ff \to \f$ is provable in $\mathbf{S4}_K$. Hence, $\F_a \Box \ff \to \F_a \f$ is provable by \eqref{a:Fmono}, which means that $\F_a \ff \to \F_a \f$ is provable by \eqref{a:FBox}. Hence, $\F_a \f \in \Gamma$. 
\end{proof}


\begin{theorem}[Completeness $\mathbf{S4}_K$]\label{thm:CompletS4K}
Let $\Delta \cup \{\f\} \subseteq \mathfrak{L}_K$. Then $\Delta \proves_{\mathbf{S4}_K} \f$ if and only if $\Delta \models \varphi$.
\end{theorem}
\begin{proof}
Soundness was established in Lemma \ref{lem:SoudnessS4K}. 
For completeness, assume $\Delta \not\proves_{\mathbf{S4}_K}\f$. Then $\Delta \cup \{\neg\varphi\}$ is consistent, and thus there is a maximal $\mathbf{S4}_K$-consistent theory $\Gamma$ such that $\Delta \cup \{\neg\varphi\} \subseteq \Gamma$ (by the standard Lindenbaum lemma, which we can use since the language is countable). By Lemma~\ref{lem:TruthLemmaS4K} we get $\bm{M}^{\mathbf{S4}_K}, \Gamma \models \Delta$ but $\bm{M}^{\mathbf{S4}_K}, \Gamma \not\models \f$. Since we showed that $\bm{M}^{\mathbf{S4}_K}$ is a $K$-model in Lemma~\ref{lem:CanModelS4K}, this witnesses $\Delta \not\models \varphi$.  
\end{proof}

Next, we present an extension of this logic which characterizes seats which are \emph{strong} and \emph{bounded} (Definition~\ref{def:SpecialSeats}).   

\begin{definition}
Let $\mathbf{S4sb}_{K}$ be the extension of $\mathbf{S4}_K$ with:
\begin{gather}
\F_{a\oplus b} \f \to \F_a \f \land \F_b \f \label{a:Fplus2}\\
\F_\1 \top \label{a:Top1}\\
\F_\0 \bot \label{a:Bot0}
\end{gather}
We write $\proves_{\mathbf{S4sb}_{K}}$ for the corresponding provability relation. 
\end{definition}

To prove strong completeness of $\mathbf{S4sb}_K$ with respect to the class of strong bounded seats (denoted by $\mathsf{sb}$), we use the following characterization result.

\begin{lemma}\label{lem:CharactS4sb}
    A $K$-seat validates axioms \eqref{a:Fplus2} -- \eqref{a:Bot0} if and only if it is strong and bounded.
\end{lemma}

\begin{proof}[Proof sketch]
Showing that strong bounded seats validate the axioms is straightforward by definitions. For the converse direction, if $\langle X, \topo, \ann \rangle$ is not strong, there are $a,b \in K$, $U \in \topo$ and $x \in X$ such that $a\oplus b \in \ann(U,x)$ but $a \notin \ann(U,x)$. Let $\bm{M} = \langle X, \topo, \ann, \val \rangle$ be any model with $\val(p) = U$. Then $\bm{M}, x \models \F_{a\oplus b} p$ but $\bm{M},x \not\models \F_a p$, showing that \eqref{a:Fplus2} is not valid. The converse direction for boundedness is proved similarly using axioms \eqref{a:Top1} and \eqref{a:Bot0}; see Appendix~\ref{app:charactS4sb}. 
\end{proof}
\begin{theorem}[Completeness $\mathbf{S4sb}_K$]\label{thm:CompletS4sbK}
Let $\Delta \cup \{\f\} \subseteq \mathfrak{L}_K$. Then $\Delta \proves_{\mathbf{S4sb}_K} \f$ if and only if $\Delta \models_{\mathsf{sb}} \varphi$.
\end{theorem}
\begin{proof}[Proof Sketch]
Soundness follows from (the `if' part of) Lemma~\ref{lem:CharactS4sb}. Completeness is shown using the canonical model $\bm{M}^{\mathbf{S4sb}}$, defined analogously to $\bm{M}^{\mathbf{S4}_K}$ (cf.\ Definition~\ref{def:CanModelS4K}), except that $X$ is the set of all maximal $\mathbf{S4sb}_K$-consistent theories. Using axioms \eqref{a:Fplus2} -- \eqref{a:Bot0} is is straightforward to see that $\bm{M}^{\mathbf{S4sb}}$ is strong and bounded.    
\end{proof}
%
%

We now present the corresponding completeness result. 

\begin{theorem}[Completeness $\mathbf{S4sb}_{K}$]\label{thm:CompletS4sbK}
Let $\Delta \cup \{\f\} \subseteq \mathfrak{L}_K$. Then $\Delta \proves_{\mathbf{S4sb}_K} \f$ if and only if $\Delta \models_{\mathsf{sb}} \varphi$.
\end{theorem}

\begin{proof}
The definition of canonical $\mathbf{S4sb}_K$-model $\bm{M}^{\mathbf{S4sb}_K}$ is the same as for the canonical $\mathbf{S4}_K$-model, except for $X$, which is the set of all $\mathbf{S4sb}_K$-consistent theories. It follows from from Lemma \ref{lem:CanModelS4K}  that $\bm{M}^{\mathbf{S4sb}_K}$ is a $K$-model  and from Lemma \ref{lem:CharactS4sb} (left to right) that it is strong and bounded. As $\mathbf{S4sb}_K$ is a conservative extension of $\mathbf{S4}_K$  the Lemma \ref{lem:TruthLemmaS4K} holds for the latter too. Soundness is the right to left direction of Lemma \ref{lem:CharactS4sb}. The rest of the proof proceeds analogously to that of Theorem \ref{thm:CompletS4K}.
\end{proof}

%

We now extend this axiomatization further in order to characterize \emph{uniform} strong bounded seats (Definition~\ref{def:SpecialSeats}).    

\begin{definition}
Let $\mathbf{S4sub}_K$ extend $\mathbf{S4sb}_K$ by the axioms: 
\begin{align}
\F_a \f \to \F_\1 \F_a \f \qquad & \neg \F_a \f \to \F_\1 \neg \F_a \f \label{au:0}\\
\F_a \f \to \Box \F_a \f \qquad & \neg \F_a \f \to \Box \neg \F_a \f \label{au:Box} 
\end{align}
We write $\proves_{\mathbf{S4sub}_{K}}$ for the corresponding derivability relation.
\end{definition}

We first show that $\mathbf{S4sub}_K$ is sound with respect to the class of strong uniform bounded $K$-seats (denoted by $\mathsf{sub})$.
\begin{lemma}[Soundness $\mathbf{S4sub}_K$]\label{lem:SoundnessS4uK}
Let $\Delta \cup \{ \f \} \subseteq \mathfrak{L}_K$. Then $\Delta\proves_{\mathbf{S4sub}_K}\f\Rightarrow\Delta \models_{\mathsf{sub}} \f$.
\end{lemma}
\begin{proof}
We show validity of the formulas $\F_a \f \to \F_\1 \F_a \f$ and $\F_a \f \to \Box \F_a \f$ of (\ref{au:0}) and (\ref{au:Box}), respectively (the proofs for the remaining ones are similar). Let $\bm{M} = \langle X, \topo, \ann_K, \val \rangle$ be a uniform $K$-model and suppose $\bm{M}, x \models \F_a\f$. Then $\exists U \in \evid_a(x)$ with $U \subseteq \sem{\f}_{\bm{M}}$. By uniformity, $U \in \evid_a(y)$ for all $y \in X$, and so $\sem{\F_a\f}_{\bm{M}} = X$. Thus, by definition, we get $\sem{\F_\1 \F_a \f}_{\bm{M}} = \For_\1(\sem{\F_a\f}_{\bm{M}}) = \For_\1(X) = X$ (since $\bm{M}$ is $\1$-bounded) and similarly $\sem{\Box \F_a \f}_{\bm{M}} = \Int(X) = X$. In particular, $\bm{M}, x \models \F_\1\F_a\f$ and $\bm{M}, x \models \Box \F_a\f$ as desired.
\end{proof}

Let $\bm{M}^{\mathbf{S4sub}_K} = \langle X, \topo, \ann_K, \val\rangle$ be the canonical model defined as before, except for using $\mathbf{S4sub}_K$-theories. It is shown exactly as in the proof of Lemmas~\ref{lem:CanModelS4K} \& ~\ref{lem:CharactS4sb} that this is a strong bounded $K$-model.

However, $\bm{M}^{\mathbf{S4sub}_K}$ need not be uniform. Therefore, to prove completeness of $\mathbf{S4sub}_K$ we further `modify' this model relative to a fixed maximal consistent theory as follows (recall that $| \f |$ denotes the set of $\Gamma \in X$ with $\f \in \Gamma$).       

\begin{definition}\label{def:CanMod-S4SUB}
Let $\Lambda \in X$ be a maximal $\mathbf{S4sub}_K$-consistent theory. Firstly, we define (with $\pm \F_a\f \in \{ \F_a\f, \neg \F_a\f \}$): 
\[X^{\Lambda} = \bigcap_{\pm\F_a \f \in \Lambda} | \pm\F_a \f | \]

The \emph{canonical $\mathbf{S4sub}_K$-model} $\bm{M}^{\Lambda} = \langle X^{\Lambda}, \topo^{\Lambda}, \ann_K^{\Lambda}, \val^{\Lambda}\rangle$ is defined as follows:
\begin{itemize}
\item $\topo^{\Lambda}$ is the subspace topology on $X^{\Lambda}$ induced by $\topo$;\footnote{That is, $V \in \topo^\Lambda$ iff $\exists U \in \topo \colon U \cap X^\Lambda = V$.}
\item $\ann_K^\Lambda(V, \Gamma) = \bigcup\{\ann_K(U, \Gamma) \mid U \cap X^\Lambda \subseteq V\}$;
\item $\val^{\Lambda}(p) = |p| \cap X^{\Lambda}$.
\end{itemize}
We write $|\f|^{\Lambda}$ instead of $|\f| \cap X^{\Lambda}$.
\end{definition}

Working with these models, the following will be useful.

\begin{lemma}\label{lem:useful}
Let $\M \in \{ \Box \} \cup \{ \F_a \mid a \in K \}$ and $\f, \ff \in \mathfrak{L}_K$. Then
$(|\f|^{\Lambda} \subseteq |\ff|) \,\Rightarrow\, (|\M \f|^{\Lambda} \subseteq |\M \ff|^{\Lambda})$.
\end{lemma}
\begin{proof}
Straightforward derivation (Appendix~\ref{app:useful}).
\end{proof}
With this, we show that $\evid_a^\Lambda$ only depends on $\Lambda$, and thus that the model $\bm{M}^\Lambda$ is uniform.   

\begin{lemma}\label{lem:CompletS4uK-help}
For any $a \in K$, $\Gamma \in X^{\Lambda}$, it holds in $\bm{M}^\Lambda$ that 
$$
\evid_a^{\Lambda}(\Gamma) =
 \{ V \in \topo^\Lambda \mid \exists \f\colon \F_a \f \in \Lambda \And |\Box \f|^{\Lambda} \subseteq V \}.
$$ 
\end{lemma}
\begin{proof}
The $\subseteq$-inclusion is clear by definition (note that $\F_a\f \in \Gamma$ iff $\F_a\f \in \Lambda$). The $\supseteq$-inclusion is established as follows. The assumption $V \in \topo^{\Lambda}$ entails that $\exists U \in \topo$ such that $|\Box \f|^{\Lambda} \subseteq U$, and so by compactness, that $| \Box \f|^{\Lambda} \subseteq |\chi| \subseteq U$ where $\chi$ is a disjunction of $\Box \chi_i$. By Lemma~\ref{lem:useful}, $|\F_a \Box \f|^{\Lambda} \subseteq |\F_a \chi|^{\Lambda}$, whence $|\F_a \f|^{\Lambda} \subseteq |\F_a \chi|^{\Lambda}$ by \eqref{a:FBox}. Since $\F_a \f \in \Gamma$, we obtain $\F_a \chi \in \Gamma$. Now $|\chi| \subseteq U$ and so $|\Box \chi| \subseteq U$. It follows that $U \in \evid_a(\Gamma)$ and so $V \in \evid_a^{\Lambda}(\Gamma)$.
\end{proof}

We show that $\bm{M}^{\Lambda}$ is indeed a canonical model for $\mathbf{S4sub}_{K}$. 

\begin{lemma}\label{lem:CanModelS4uK}
$\bm{M}^{\Lambda}$ is a strong uniform bounded $K$-model.
\end{lemma}
\begin{proof}[Proof Sketch]
Uniformity follows from Lemma~\ref{lem:CompletS4uK-help}; and $\langle X^\Lambda, \topo^\Lambda, \ann^\Lambda_K \rangle$ is shown strong and bounded similarly to Lemmas~\ref{lem:CanModelS4K} \& \ref{lem:CharactS4sb}. For instance, \eqref{A(iv)} is shown as follows (see Appendix~\ref{app:CanModelS4uK} for more details). 
If $a \in \ann_K^{\Lambda}(U)$ and $b \in \ann_K^{\Lambda}(V)$ for $U, V \in \topo^{\Lambda}$, then $a \in \ann_K(U', \Lambda)$ and $b \in \ann_K(V', \Lambda)$ for some $U', V' \in \topo$ such that $U' \cap X^{\Lambda} \subseteq U$ and $V' \cap X^{\Lambda} \subseteq V$ (Lemma \ref{lem:CompletS4uK-help}). We know that $ab, ba \in \ann_K(U' \cap V', \Lambda)$ and $U' \cap V' \cap X^{\Lambda} \subseteq U \cap V$. Hence, $ab, ba \in \ann_K^{\Lambda}(U \cap V)$.
\end{proof}

\begin{lemma}[Truth Lemma $\bm{M}^\Lambda$]\label{lem:TruthLemmaS4uK}
Let $\chi \in \mathfrak{L}_K$ and $\Lambda$ be a maximal $\mathbf{S4sub}_K$-consistent set. Then $|\chi|^{\Lambda} = \llbracket \chi \rrbracket_{\bm{M}^{\Lambda}}$.
\end{lemma}

\begin{proof}
Structural induction on $\chi$. The base case $\chi \in \Prop$ 
and inductive cases for $\chi \in \{ \neg \f, \f \land \ff\}$ are straightforward. 

\emph{Case} $\chi = \Box\f$: If $\Box \f \in \Gamma \in X^\Lambda$, then $\Gamma \in |\Box\f|^\Lambda \subseteq |\f|^\Lambda$. By the induction hypothesis, $|\f|^\Lambda = \llbracket \f \rrbracket_{\bm{M}^\Lambda}$, and by definition of $\topo^\Lambda$, $|\Box\f|^\Lambda \in \topo^\Lambda$, implying $\bm{M}^\Lambda, \Gamma \models \Box \f$.

For the converse, suppose $\Box\f \notin \Gamma$ but $\Gamma \in \llbracket \Box\f \rrbracket_{\bm{M}^\Lambda}$. That is, $\exists U \in \topo$ such that 
$\Gamma \in U \cap X^\Lambda \subseteq |\f|^\Lambda$. By definition of  $\topo$, 
$U = \bigcup_{i \in I} |\Box \psi_i|$ for some set of formulas $\{\psi_i\}_{i \in I}$. Without loss of generality, suppose $\Gamma \in |\Box \psi_0|$. Since 
$|\Box \psi_0|^\Lambda \subseteq |\f|^\Lambda$, we can use Lemma~\ref{lem:useful} and $\mathbf{S4}$ to obtain 
$|\Box \psi_0|^\Lambda \subseteq |\Box\f|^\Lambda$, which yields $\Gamma \in |\Box\f|^{\Lambda}$, a contradiction.

\emph{Case} $\chi = \F_a\f$: If $\F_a\f \in \Gamma \in X^\Lambda$, then $|\Box\f|^\Lambda \in \evid_a^\Lambda(\Gamma)$ (by the definition of $\evid_a^\Lambda$) and 
$|\Box\f|^\Lambda \subseteq |\f|^\Lambda = \llbracket \f \rrbracket_{\bm{M}^\Lambda}$ (by $\mathbf{S4}$ and the induction hypothesis), implying  $\bm{M}^\Lambda, \Gamma \models \F_a\f$.

For the converse, suppose $\F_a\f \notin \Gamma$ but $\Gamma \in \llbracket \F_a\f \rrbracket_{\bm{M}^\Lambda}$. Then $\exists U \in \topo$ such that 
$U \cap X^\Lambda \subseteq |\f|^\Lambda$ and $U \cap X^\Lambda \in \evid_a^\Lambda(\Gamma)$. By Lemma~\ref{lem:CompletS4uK-help}, the latter implies that 
$\exists\F_a\psi \in \Gamma$ with $|\Box\psi|^\Lambda \subseteq U \cap X^\Lambda$. Thus, we have 
$|\Box\psi|^\Lambda \subseteq |\f|^\Lambda$, so we can use Lemma~\ref{lem:useful} and \eqref{a:FBox} to obtain 
$|\F_a\psi|^\Lambda \subseteq |\F_a\f|^\Lambda$, which implies $\Gamma \in |\F_a \f|$, a contradiction.
\end{proof}

We are now ready to prove strong completeness of $\mathbf{S4sub}_K$. 

\begin{theorem}[Completeness $\mathbf{S4sub}_K$]\label{thm:CompletS4uK}
Let $\Delta \cup \{\f\} \subseteq \mathfrak{L}_K$. Then $\Delta \proves_{\mathbf{S4sub}_K} \f$ if and only if $\Delta \models_{\mathsf{sub}} \varphi$.
\end{theorem}
\begin{proof}
Soundness was established in Lemma~\ref{lem:SoundnessS4uK}. For completeness, assume that $\Delta \not\proves_{\mathbf{S4sub}_K}\f$. Then there exists a maximal $\mathbf{S4sub}_K$-consistent theory $\Lambda$ with $\Delta \cup \{\neg\f\} \subseteq \Lambda$. Now $\bm{M}^\Lambda$ is a strong uniform bounded $K$-model (Lemma~\ref{lem:CanModelS4uK}) with $\bm{M}^\Lambda, \Lambda \models \Delta$ but $\bm{M}^\Lambda, \Lambda \not\models \f$ (Lemma~\ref{lem:TruthLemmaS4uK}).   
\end{proof}

\section{Expansions with the global modality}\label{sec:Logics and completeness II}
In this section, we extend the signature of $\mathfrak{L}_K$ with the global modality denoted by $\A$; let $\mathfrak{L}_{K\forall}$ be the resulting expanded language. Expanding the language with $\A$ allows us to define modal operators expressing $\Belop^a_b$ and $\Knop^a_b$ of Definition~\ref{def:BelAndKn} (cf.~Lemma \ref{lem:belief knowledge}). Moreover, it allows us to characterize the class of uniform strong bounded seats 
(cf.~Lemma \ref{lem:uniform define}), which is not possible in $\mathfrak{L}_K$ alone (cf.~Corollary~\ref{cor:disjointUnionExample}).

 The global modality $\A$ is interpreted in the standard manner. For a model $\bm{M} = \langle X, \topo, \ann_K, \val \rangle$ , formula  $\f \in \mathfrak{L}_{K\forall}$, and  state $x \in X$,  
$$
\bm{M}, x \models \A\f \quad \text{iff} \quad \text{for all } y \in X,\; \bm{M}, y \models \f.
$$
Intuitively, the formula $\A\f$ states that $\f$ is true \emph{globally}, that is, at all states. As usual, we define the global existential modality $\E$ by $\E\f \coloneq \neg \A \neg \f$.

\begin{remark}\label{rem:GlobalDefinability}
Let $\bm{M} = \langle X, \topo, \ann_K, \val \rangle$ be a  $K$-model.
Suppose   $\evid_\1(x) = \{X\}$ for some $x \in X$. That is, $X$ is the only piece of  evidence with a cost of $\1$.
Then, for every formula $\f \in \mathfrak{L}_{K\forall}$,
\[
\bm{M}, x \models \A \f \quad \text{iff} \quad \bm{M}, x \models \Box_\1 \f.
\]
Consequently, in any model satisfying this condition for all $x\in X$, the global modality $\A$ is definable in the logic $\mathbf{S4}_K$.
A similar observation in the context of \emph{stratified evidence models} was made in \cite[Theorem~6.1]{BalbianiEtAl2019a}.
However, since we do not make this assumption, the global modality can not be defined in this way in our framework. In fact, later on, using the appropriate notion of bisimulation, we show that $\A$ is not definable by any formula in $\mathfrak{L}_{K}$ (Corollary~\ref{cor:bisimul example}). 
\end{remark}

Let $\mathbf{S4sb}_{K\forall}$ extend $\mathbf{S4sb}_K$ with the axioms
\begin{gather}
	\A\f \to \f \label{ax:Areflex}\\
	\A\f \to \A\A\f \label{ax:Atrans}\\
    \f \to \A\E\f \label{ax:Asymm}\\
    \A\f \wedge \A\psi \to \A(\f \wedge \psi) \label{ax:Anorm}\\
    \A\f \to \Box \f \wedge \F_\1 \f \label{ax:ABox}
\end{gather}
and the $\A$-necessitation rule $\f / \A \f$.
We write $\proves_{\mathbf{S4sb}_{K\forall}}$ for the corresponding derivability relation.

As before, we first show that $\mathbf{S4sb}_{K\forall}$ is sound with respect to the class of strong bounded $K$-seats.
\begin{lemma}[Soundness $\mathbf{S4sb}_{K\forall}$]\label{lem:SoundnessS4sbKforall}
Let $\Delta \cup \{\f\} \subseteq \mathfrak{L}_{K\forall}$. Then $\Delta\proves_{\mathbf{S4sb}_{K\forall}}\f \implies \Delta \models_{\mathsf{sb}} \f$.
\end{lemma}
\begin{proof}
Immediate from the soundness of $\mathbf{S4sb}_K$ (Lemmas~\ref{lem:SoudnessS4K}, \ref{lem:CharactS4sb}) and the semantics of the global modality $\A$. 
\end{proof}

We now show that $\mathbf{S4sb}_{K\forall}$ is also complete for $\mathfrak{L}_{K\forall}$  interpreted over strong bounded seats via the canonical model construction. Let $X$ be the set of all maximal consistent $\mathbf{S4sb}_{K\forall}$-theories. For any $\Lambda \in X$, we define the canonical model $\bm{M}^{\Lambda}= \langle X^{\Lambda}, \topo^{\Lambda}, \ann_K^{\Lambda}, \val^{\Lambda}\rangle$ analogously to the canonical $\mathbf{S4sub}_{K}$-model (cf.~Definition~\ref{def:CanMod-S4SUB}), except that in this case,  the domain of the model is 
\[
X^{\Lambda} = \bigcap_{\A\f \in \Lambda} |\f|.
\]
It can be shown  exactly as in the proofs of Lemmas~\ref{lem:CanModelS4K} \& ~\ref{lem:CanModelS4uK} that this is a strong bounded $K$-model.

Similarly to Lemma~\ref{lem:useful}, the following will be useful when working with these models. 

\begin{lemma}\label{lem:useful global}
Let $\M \in \{ \Box, \A \} \cup \{ \F_a \mid a \in K \}$ and $\f, \ff \in \mathfrak{L}_{K\forall}$.
Then $(|\f|^{\Lambda} \subseteq |\ff|) \implies (|\M \f|^{\Lambda} \subseteq |\M \ff|^{\Lambda})$.
\end{lemma}

\begin{proof}
Straightforward derivation (Appendix~\ref{app:useful global}).
\end{proof}

The following lemma shows that $\bm{M}^{\Lambda}$ is indeed a canonical model for $\mathbf{S4sb}_{K\forall}$. 
\begin{lemma}[Truth Lemma $\mathbf{S4sb}_{K\forall}$]\label{lem:truth S4sbKforall}
$\bm{M}^{\Lambda}$ is a strong and bounded $K$-model and 
$|\chi|^{\Lambda} = \llbracket \chi \rrbracket_{\bm{M}^{\Lambda}}$ for every $\chi \in \mathfrak{L}_{K\forall}$.
\end{lemma}

\begin{proof}
The proof proceeds analogously to Lemma~\ref{lem:TruthLemmaS4uK}. Full details are provided in Appendix~\ref{app:TruthLemmaAll}.
\end{proof}

We are now ready to prove strong completeness for $\mathbf{S4sb}_{K\forall}$ with respect to all strong bounded $K$-seats.

\begin{theorem}[Completeness $\mathbf{S4sb}_{K\forall}$]\label{thm:CompletS4sbKall}
Let $\Delta \cup \{\f\} \subseteq \mathfrak{L}_{K\forall}$. Then $\Delta \proves_{\mathbf{S4sb}_{K\forall}} \f$ if and only if $\Delta \models_{\mathsf{sb}} \varphi$.
\end{theorem}
\begin{proof}
Soundness was shown in Lemma~\ref{lem:SoundnessS4sbKforall}. For completeness, let $\Delta \not\proves_{\mathbf{S4sb}_{K\forall}}\f$. Then there exists a maximal consistent $\Lambda\in X$ with $\Delta \cup \{\neg\f\} \subseteq \Lambda$ and $\bm{M}^\Lambda$ is a strong bounded $K$-model in which $\Delta \models \f$ is not valid (Lemma~\ref{lem:truth S4sbKforall}).   
\end{proof}

Next, we present an extension of this logic that characterizes \emph{uniform} seats. 

\begin{lemma}\label{lem:uniform define}
Let $\bm{S}= \langle X, \topo, \ann_K\rangle$ be a strong  bounded $K$-seat. Then, the following are equivalent.
\begin{enumerate}
\item[(i)] $\bm{S}$ is uniform.
    \item[(ii)] $\bm{S}$ validates the following axiom for every $a \in K$: 
\begin{gather}
    \F_a\f \rightarrow \A\F_a\f \label{ax:uniformity global}
\end{gather}
  \item[(iii)]  $\bm{S}$ validates the following axiom for every $a \in K$: 
\begin{gather}
    \neg \F_a\f \rightarrow \A\neg \F_a\f \label{ax:uniformity global 2}
\end{gather}
\end{enumerate}
\end{lemma}
\begin{proof}
It is straightforward to check that if $\bm{S}$ is uniform, then it validates the above axioms. Conversely, suppose $\bm{S}$ is not uniform. Then, there exist $x_1, x_2 \in X$, $U \in \topo$, and  $a \in K$ such that $U \in \evid_a(x_1)$, but $U \not \in \evid_a(x_2)$. Let $\val$ be a valuation with $\val(p)=U$ for some proposition $p$. Then for $\bm{M}=\langle \bm{S}, \val \rangle$, we have $\bm{M},x_1 \models \F_a p$, but  $\bm{M},x_1 \not \models \A\F_a p $, showing that $\bm{S}$ does not validate axiom \eqref{ax:uniformity global}. Similarly, $\bm{M},x_2 \models \neg \F_a p  $ but  $\bm{M},x_2 \not \models \neg\A\F_a p$ showing $\bm{S}$ does not validate axiom \eqref{ax:uniformity global 2}.
\end{proof}

As a corollary, we obtain the following. Let $\mathbf{S4sub}_{K\forall}$ be the extension of $\mathbf{S4sb}_{K\forall}$ with axioms \eqref{ax:uniformity global}  and  \eqref{ax:uniformity global 2}.

\begin{theorem}[Completeness $\mathbf{S4sub}_{K\forall}$]\label{thm:characterize uniform}
If $\Delta \cup \{\f\} \subseteq \mathfrak{L}_{K\forall}$, then $\Delta \proves_{\mathbf{S4sub}_{K\forall}} \f$ if and only if $\Delta \models_{\mathsf{sub}} \varphi$. 
\end{theorem}
\begin{proof}
See Appendix~\ref{app:characterize uniform}. 
\end{proof}

We define the logical formulas to characterize (weighted) justified belief and knowledge operators as follows. 
\[
\Bel^a_b \f \coloneq \A\Diamond_b\Box_a \f \quad \text{and}  \quad \Kn^a_b \f\coloneq  \Box_a\f \wedge \Bel^a_b \f. 
\]
The following lemma states that the modalities $\Bel^a_b $ and  $\Kn^a_b$ indeed capture the intuitive weighted generalizations of the belief and knowledge operators  discussed in Definition~\ref{def:BelAndKn}.
\begin{lemma}\label{lem:belief knowledge}
Let $\bm{M} = \langle X, \topo, \ann_K, \val \rangle$ be any $\mathbf{S4sub}_{K}$-model. For any formula $\f \in \mathfrak{L}_{K\forall}$, and state $x \in X$, 
\begin{align*}
\bm{M}, x \models \Bel^a_b \f \quad &\text{iff} \quad  x \in  \Belop^a_b (\semM{\f}), \\ 
\bm{M}, x \models \Kn^a_b \f \quad &\text{iff} \quad  x \in  \Knop^a_b (\semM{\f}). 
\end{align*}
\end{lemma}
\begin{proof}
See Appendix~\ref{app:belief knowledge}.
\end{proof}
\begin{remark}
Note that, in the case $a=b=\0$, the formulas $\Bel^a_b \f$ and $\Kn^a_b \f$ reduce to $\A\Diamond\Box\f$ and $\A\Diamond\Box\f \wedge \Box \f$, respectively, which define the standard belief ($\mathrm{B}$) and knowledge ($\mathrm{K}$) operators in TEL, as established in \cite[Proposition 6]{BaltagEtAl2022}. 
\end{remark}

\section{Modal undefinabiliy}\label{sec:Undefinability}

In the previous sections, we showed that the class of \emph{strong bounded} seats is modally definable over $\mathfrak{L}_K$ (Lemma~\ref{lem:CharactS4sb}) and that the class of \emph{uniform} seats is modally definable over the expanded language $\mathfrak{L}_{K\forall}$ (Lemma~\ref{lem:uniform define}). 

In this section, we prove \emph{undefinability} results for some classes of seats introduced in Section~\ref{sec:seats}. We show that the class of \emph{uniform seats} is not definable in $\mathfrak{L}_K$ and that the global modality is not definable in the language $\mathfrak{L}_{K}$.

As usual, we say that a class of $K$-seats $\mathsf{C}$ is \emph{modally definable over $\mathfrak{L} \in \{ \mathfrak{L}_K, \mathfrak{L}_{K\forall}\}$} if there exists a formula $\f \in \mathfrak{L}$ such that $\bm{S} \models \f \Leftrightarrow \bm{S} \in \mathsf{C}$ holds for all $K$-seats $\bm{S}$. 

To show that the class of uniform seats is not modally definable over $\mathfrak{L}_K$, we make use of the following construction. 

\begin{definition}\label{def:DisjointUnion}
Let $\{ \bm{M}_i \}_{i\in I}$ be an $I$-indexed collection of $K$-models $\bm{M}_i = \langle X_i, \topo_i, \ann_i, \val_i  \rangle$. We define the \emph{disjoint union} 
$\biguplus_{i\in I} \bm{M}_i = \langle X, \topo, \ann, \val  \rangle$
as follows:
\begin{itemize}
    \item $X$ is the disjoint union $\biguplus_{i\in I} X_i$;
    \item $\topo$ is the usual disjoint union topology;\footnote{That is, $U \in \topo$ iff $\forall i \in I\colon U \cap X_i \in \topo_i$.} 
    \item $\ann(U, x) = \ann_i(X_i \cap U, x)$, where $x \in X_i$;
    \item $x \in \val(p) \Leftrightarrow x \in \val_i(p)$, where $x \in X_i$. 
\end{itemize}
\end{definition}

Similarly to standard modal logic, we obtain the following preservation result with respect to formulas in $\mathfrak{L}_K$ (note that the fact below does not hold for formulas $\A\f$).  

\begin{proposition}\label{prop:DisjointUnion}
Let $\bm{M} = \biguplus_{i\in I} \bm{M}_i$ be a disjoint union as above. Then $\sem{\chi}_{\bm{M}_i} = \semM{\chi} \cap X_i$ for all $\chi \in \mathfrak{L}_K, i \in I$. 
\end{proposition}
\begin{proof}
Simple induction on $\chi$, see Appendix~\ref{app:DisjointUnion}. 
\end{proof}

Therefore, a class $\mathsf{C}$ is definable over $\mathfrak{L}_K$ only if it is \emph{closed under disjoint unions}, i.e.\ $\{ \bm{S}_i \}_I \subseteq \mathsf{C} \Rightarrow \biguplus_{I} \bm{S}_i \in \mathsf{C}$. With this, it is straightforward to show the following undefinability results (for the context of 2. below, recall Remark~\ref{rem:GlobalDefinability}). 

\begin{corollary}[Undefinability I]\label{cor:disjointUnionExample}
\phantom{a}
\begin{enumerate}
    \item The class of uniform $K$-seats is not definable in $\mathfrak{L}_K$.
    \item The class of $K$-seats satisfying  $\forall x\in X \colon \evid_\1(X) = \{ x \}$ is not definable in $\mathfrak{L}_K$.  
\end{enumerate}
\end{corollary}

%

Next, we define the notion of \emph{bisimulation} for $K$-models. 

\begin{definition}\label{def:bisimulation}
    Let $\bm{M}_1 = \langle X_1, \topo_1, \ann_1, \val_1  \rangle$ and $\bm{M}_2 = \langle X_2, \topo_2,\allowbreak \ann_2, \val_2  \rangle$ be $K$-models. A {\em bisimulation} between these is a relation $Z \subseteq X_1 \times X_2$, such that $x_1 Z x_2$ implies:

\begin{enumerate}
    \item[(i)]   $x_1 \in \val_1(p)$ iff  $x_2 \in \val_2(p)$ for all $p \in \Prop$.
    \item[(ii)]  If $x_1 \in U_1 \in  \topo_1$, then $\exists U_2 \in \topo_2$  s.t.~$x_2 \in U_2$ and for any $y_2 \in U_2$ there exists $y_1 \in U_1$ s.t.~$y_1 Z y_2$.
      \item[(iii)]  If $x_2 \in U_2 \in  \topo_2$, then $\exists U_1 \in \topo_1$  s.t.~$x_1 \in U_1$ and for any $y_1 \in U_1$ there exists $y_2 \in U_2$ s.t.~$y_1 Z y_2$.
      \item[(iv)]  For any $a \in K$,  if $U_1 \in  {\evid_1}_a(x_1)$, then $\exists U_2 \in {\evid_2}_a(x_2)$  s.t.~for any $y_2 \in U_2$ there exists $y_1 \in U_1$ s.t.~$y_1 Z y_2$.
       \item[(v)] For any $a \in K$,  if $U_2 \in  {\evid_2}_a(x_2)$, then $\exists U_1 \in {\evid_1}_a(x_1)$  s.t.~for any $y_1 \in U_1$ there exists $y_2 \in U_2$ s.t.~$y_1 Z y_2$.
\end{enumerate}  
A bisimulation is {\em global} if it further satisfies the following:
\begin{enumerate}
    \item[(vi)] For every $x_1 \in X_1$, there exists $x_2 \in X_2$ s.t.~$x_1 Z x_2$,
     \item[(vii)] For every $x_2 \in X_2$, there exists $x_1 \in X_1$ s.t.~$x_1 Z x_2$.
\end{enumerate}
\end{definition}

\begin{theorem}\label{thm:bisimulation}
Let $\bm{M}_1 = \langle X_1, \topo_1, \ann_1, \val_1  \rangle$ and $\bm{M}_2 = \langle X_2, \topo_2,\allowbreak \ann_2, \val_2  \rangle$  be $K$-models, and $Z \subseteq X_1 \times X_2$ be a bisimulation between them. For any $x_1 \in X_1$, $x_2 \in X_2$ such that $x_1 Z x_2$, and formula $\f \in \mathfrak{L}_{K}$, 
\[
\bm{M}_1, x_1 \models \f \quad \text{iff} \quad \bm{M}_2, x_2 \models \f.
\]
Moreover, if $Z$ is global, the equivalence holds for all $\f \in \mathfrak{L}_{K\forall}$. 
\end{theorem}
\begin{proof}
See Appendix \ref{app:Proof of Theorem bisimulation}. 
\end{proof}
As an application of this theorem, we can now show the following undefinability results. 

\begin{corollary}[Undefinability II]\label{cor:bisimul example}
\phantom{a}
\begin{enumerate}
    \item The global modality $\A$ is not definable in $\mathfrak{L}_{K}$. 
    \item The cost seat property $\forall U \in \topo\colon \bigsqcup \ann_x(U) \in \ann_x(U)$ at $x$ is not (locally) definable in  $\mathfrak{L}_{K\forall}$.
\end{enumerate}
\end{corollary}
\begin{proof}
See Appendix~\ref{app:corr example}. 
\end{proof}

\section{Related Work}\label{sec:RelatedWork}

%
Several works have explored epistemic logics that explicitly represent the cost of information and the resources needed to obtain it; see  \cite{NaumovTao2015,DolgorukovEtAl2024,GalmicheEtAl2019,BelardinelliRendsvig2021,Solaki2022,CostantiniEtAl2021,BalbianiEtAl2019a,BaurStuder2021,HeifetzMongin2001} for example. 
However, our approach differs from the existing work in two fundamental ways:

(i) In existing work, costs, budgets and resources are often represented using a specific structure, often a numerical one such as natural or rational numbers, or an unspecified cardinal number \cite{NaumovTao2015,DolgorukovEtAl2024,BelardinelliRendsvig2021,Solaki2022,BalbianiEtAl2019a,HeifetzMongin2001}. Moreover, these works usually assume that costs are uniform, i.e.~they do not depend on the state.  In contrast, our framework is \emph{more general}, allowing for non-uniform costs and arbitrary semirings. This enables the modelling of non-linear resource structures, such as security clearance levels or multidimensional resource vectors (e.g.\ time \emph{and} money), the values of which cannot be adequately captured by a linear structure.

(ii) Existing work typically uses discrete semantic structures coming from modal logic, knowledge representation, and multi-agent systems, such as Kripke frames, neighbourhood frames, and combinations of the two. While these can be seen as endowed with a discrete topology, such an approach misses the opportunity to connect resource-bounded epistemic logic with the work on the topology of observable properties, which has proved to be fruitful in other areas of computer science. In our work, \emph{seats provide a bridge} between topological models of evidence and  resource-conscious models of evidence access. 

An existing approach that is perhaps the closest to our contribution is the \emph{Stratified Evidence Logic} of \cite{BalbianiEtAl2019a}. It builds on Evidence Logic \cite{vanBenthemEtAl2014}, based on neighborhood models, and extends it with a representation of the cost of evidence -- the number (taken from an arbitrary cardinal) of pieces of evidence that need to be combined in order to obtain a given piece of evidence. Our framework subsumes this: a stratified evidence frame can be viewed as a strong, uniform, bounded $K$-seat where $K$ is the tropical semiring over an arbitrary cardinal.

Another closely related field is that of \emph{Markovian logics}; see \cite{HeifetzMongin2001,KozenEtAl2013,Zhou2007,Zhou2014} for example. These are modal logics for reasoning about probabilistic belief in Harsanyi's type spaces and transitions in Markov processes. The language of these logics contains modal formulas $\mathrm{L}_a \varphi$ indicating that the subjective probability of $\varphi$ for an agent is at least $a$, or that the probability of transitioning to a state where $\varphi$ holds is at least $a$. Markovian logics aim at modelling probability of propositions instead of accessibility (cost) of evidence needed to support them. Accordingly, they use specific numerical semirings. However, Example \ref{exam:BorelCost} suggests a connection to our framework. Markovian models give rise to particular seats and our formulas $\F_a \f$ can be approximated by $\mathrm{L}_{1 - a} \f$ in the Markovian setting. This shows that Markovian logics can be integrated into the broader family of logics discussed here, with seats serving as a unifying class of structures.


\section{Conclusion}\label{sec:Conclusion}

We have taken the first steps in studying Resource-Aware Topological Evidence Logic. We introduced the concept of a semiring-annotated topological space (seat), a mathematical model combining two ideas: first, evidence linked to observable properties of an environment is modelled using open sets of a topology on the possible states of the environment; and second, the resources sufficient to obtain this evidence are modelled using an annotation function from the collection of open sets to subsets of a semiring. We have also defined seat-based epistemic logics and established strong completeness for logics of specific classes of seats. Moreover, we define the notions of disjoint unions and bisimulations to delineate the expressive power of the discussed logics.

There are a number of interesting problems that we will pursue in the future. Firstly, we aim to determine the decidability and computational complexity of our seat-based logics, to see whether an explicit representation of resources incurs a computational cost. 
Secondly, we will develop dynamic extensions of our logics to model the change of costs and information dynamics. We will explore two approaches to dynamics on topological spaces. The first one is based on continuous transition functions on the state space (dynamical systems) and the second one on discrete information updates using subspace topologies (dynamic epistemic logic). 
Thirdly, drawing upon work in multi-agent TEL \cite{BaltagEtAl2022a, BaltagEtAl2025a, FernandezGonzalez2018, RamirezAbarca2015}, we will develop multi-agent variants of seat-based TEL to model situations where agents need to cooperate on their respective budgets.  
Fourthly, building on the applications of epistemic logic in learning theory \cite{BaltagEtAl2016a}, we will explore resource-aware counterparts to the notions of learnability in the limit and the solvability of inductive problems.
Finally, to better understand the mathematical properties of our framework, we will extend existing work on topological modal logics that proves completeness with respect to specific topological spaces such as $\mathbb{R}^n$ \cite{BaltagEtAl2019} and explores modal characterisation and definability \cite{tenCateEtAl2009}, to our resource-aware framework. 

\paragraph{Acknowledgement} This work was supported by Operational Programme Jan \'{A}mos Komenský project `\emph{Biography of Fake News with a Touch of AI: Dangerous Phenomenon through the Prism of Modern Human Sciences}' (reg.\ CZ.02.01.01/00/23\_025/0008724).


\begin{thebibliography}{40}
\providecommand{\natexlab}[1]{#1}
\providecommand{\url}[1]{\texttt{#1}}
\expandafter\ifx\csname urlstyle\endcsname\relax
  \providecommand{\doi}[1]{doi: #1}\else
  \providecommand{\doi}{doi: \begingroup \urlstyle{rm}\Url}\fi

\bibitem[Abramsky(1991)]{Abramsky1991}
S.~Abramsky.
\newblock Domain theory in logical form.
\newblock \emph{Annals of Pure and Applied Logic}, 51\penalty0
  (1–2):\penalty0 1--77, Mar. 1991.
\newblock ISSN 0168-0072.
\newblock \doi{10.1016/0168-0072(91)90065-t}.

\bibitem[Aiello et~al.(2003)Aiello, {van Benthem}, and
  Bezhanishvili]{AielloEtAl2003}
M.~Aiello, J.~{van Benthem}, and G.~Bezhanishvili.
\newblock Reasoning about space: {T}he modal way.
\newblock \emph{Journal of Logic and Computation}, 13\penalty0 (6):\penalty0
  889--920, Dec. 2003.
\newblock ISSN 1465-363X.
\newblock \doi{10.1093/logcom/13.6.889}.

\bibitem[Balbiani et~al.(2019)Balbiani, Fernández-Duque, Herzig, and
  Lorini]{BalbianiEtAl2019a}
P.~Balbiani, D.~Fernández-Duque, A.~Herzig, and E.~Lorini.
\newblock Stratified evidence logics.
\newblock In \emph{Proceedings of the Twenty-Eighth International Joint
  Conference on Artificial Intelligence}, IJCAI-2019, pages 1523--1529.
  International Joint Conferences on Artificial Intelligence Organization, Aug.
  2019.
\newblock \doi{10.24963/ijcai.2019/211}.

\bibitem[Baltag et~al.(2016{\natexlab{a}})Baltag, Bezhanishvili, Özgün, and
  Smets]{BaltagEtAl2016}
A.~Baltag, N.~Bezhanishvili, A.~Özgün, and S.~Smets.
\newblock Justified belief and the topology of evidence.
\newblock In \emph{Logic, Language, Information, and Computation}, pages
  83--103. Springer Berlin Heidelberg, 2016{\natexlab{a}}.
\newblock ISBN 9783662529218.
\newblock \doi{10.1007/978-3-662-52921-8_6}.

\bibitem[Baltag et~al.(2016{\natexlab{b}})Baltag, Gierasimczuk, and
  Smets]{BaltagEtAl2016a}
A.~Baltag, N.~Gierasimczuk, and S.~Smets.
\newblock On the solvability of inductive problems: {A} study in epistemic
  topology.
\newblock In \emph{{TARK 2015}}, volume 215 of \emph{Electronic Proceedings in
  Theoretical Computer Science}, pages 81--98. Open Publishing Association,
  2016{\natexlab{b}}.
\newblock \doi{10.4204/eptcs.215.7}.

\bibitem[Baltag et~al.(2019)Baltag, Bezhanishvili, and
  Fernández~González]{BaltagEtAl2019}
A.~Baltag, N.~Bezhanishvili, and S.~Fernández~González.
\newblock The {McKinsey}-{Tarski} theorem for topological evidence logics.
\newblock In R.~Iemhoff, M.~Moortgat, and R.~{de Queiroz}, editors,
  \emph{Logic, Language, Information, and Computation}, pages 177--194.
  Springer Berlin Heidelberg, 2019.
\newblock ISBN 9783662595336.
\newblock \doi{10.1007/978-3-662-59533-6_11}.

\bibitem[Baltag et~al.(2022{\natexlab{a}})Baltag, Bezhanishvili, and
  Fernández~González]{BaltagEtAl2022a}
A.~Baltag, N.~Bezhanishvili, and S.~Fernández~González.
\newblock Topological evidence logics: Multi-agent setting.
\newblock In \emph{Language, Logic, and Computation (TbiLLC 2019)}, pages
  237--257. Springer International Publishing, 2022{\natexlab{a}}.
\newblock ISBN 9783030984793.
\newblock \doi{10.1007/978-3-030-98479-3_12}.

\bibitem[Baltag et~al.(2022{\natexlab{b}})Baltag, Bezhanishvili, \"{O}zg\"{u}n,
  and Smets]{BaltagEtAl2022}
A.~Baltag, N.~Bezhanishvili, A.~\"{O}zg\"{u}n, and S.~Smets.
\newblock Justified belief, knowledge, and the topology of evidence.
\newblock \emph{Synthese}, 200\penalty0 (6), Dec. 2022{\natexlab{b}}.
\newblock ISSN 1573-0964.
\newblock \doi{10.1007/s11229-022-03967-6}.

\bibitem[Baltag et~al.(2025{\natexlab{a}})Baltag, Bezhanishvili, and
  Fernández-Duque]{BaltagEtAl2025}
A.~Baltag, N.~Bezhanishvili, and D.~Fernández-Duque.
\newblock The topology of surprise.
\newblock \emph{Artificial Intelligence}, 349:\penalty0 104423, Dec.
  2025{\natexlab{a}}.
\newblock ISSN 0004-3702.
\newblock \doi{10.1016/j.artint.2025.104423}.

\bibitem[Baltag et~al.(2025{\natexlab{b}})Baltag, Gattinger, and
  Gomes]{BaltagEtAl2025a}
A.~Baltag, M.~Gattinger, and D.~Gomes.
\newblock Virtual group knowledge and group belief in topological evidence
  models (extended version).
\newblock Technical report, 2025{\natexlab{b}}.

\bibitem[Baur and Studer(2021)]{BaurStuder2021}
M.~Baur and T.~Studer.
\newblock Semirings of evidence.
\newblock \emph{Journal of Logic and Computation}, 31\penalty0 (8):\penalty0
  2084--2106, Feb. 2021.
\newblock ISSN 1465-363X.
\newblock \doi{10.1093/logcom/exab007}.

\bibitem[Belardinelli and Rendsvig(2021)]{BelardinelliRendsvig2021}
G.~Belardinelli and R.~K. Rendsvig.
\newblock Epistemic planning with attention as a bounded resource.
\newblock In \emph{Logic, Rationality, and Interaction (LORI 2021)}, pages
  14--30. Springer International Publishing, 2021.
\newblock ISBN 9783030887087.
\newblock \doi{10.1007/978-3-030-88708-7_2}.

\bibitem[Bistarelli et~al.(1997)Bistarelli, Montanari, and
  Rossi]{BistarelliEtAl1997}
S.~Bistarelli, U.~Montanari, and F.~Rossi.
\newblock Semiring-based constraint satisfaction and optimization.
\newblock \emph{Journal of the ACM}, 44\penalty0 (2):\penalty0 201--236, Mar.
  1997.
\newblock ISSN 1557-735X.
\newblock \doi{10.1145/256303.256306}.

\bibitem[Bjorndahl and Özgün(2019)]{BjorndahlOzgun2019}
A.~Bjorndahl and A.~Özgün.
\newblock Logic and topology for knowledge, knowability, and belief.
\newblock \emph{The Review of Symbolic Logic}, 13\penalty0 (4):\penalty0
  748--775, Oct. 2019.
\newblock ISSN 1755-0211.
\newblock \doi{10.1017/s1755020319000509}.

\bibitem[Blackburn et~al.(2001)Blackburn, {de Rijke}, and
  Venema]{BlackburnEtAl2001}
P.~Blackburn, M.~{de Rijke}, and Y.~Venema.
\newblock \emph{{Modal Logic}}.
\newblock Cambridge University Press, Cambridge, 2001.
\newblock \doi{10.1017/cbo9781107050884}.

\bibitem[Bonatti et~al.(2002)Bonatti, De~Capitani~di Vimercati, and
  Samarati]{BonattiEtAl2002}
P.~Bonatti, S.~De~Capitani~di Vimercati, and P.~Samarati.
\newblock An algebra for composing access control policies.
\newblock \emph{ACM Transactions on Information and System Security},
  5\penalty0 (1):\penalty0 1--35, Feb. 2002.
\newblock ISSN 1557-7406.
\newblock \doi{10.1145/504909.504910}.

\bibitem[Costantini et~al.(2021)Costantini, Formisano, and
  Pitoni]{CostantiniEtAl2021}
S.~Costantini, A.~Formisano, and V.~Pitoni.
\newblock An epistemic logic for multi-agent systems with budget and costs.
\newblock In \emph{Logics in Artificial Intelligence (JELIA` 2021)}, pages
  101--115. Springer International Publishing, 2021.
\newblock ISBN 9783030757755.
\newblock \doi{10.1007/978-3-030-75775-5_8}.

\bibitem[Dolgorukov et~al.(2024)Dolgorukov, Galimullin, and
  Gladyshev]{DolgorukovEtAl2024}
V.~Dolgorukov, R.~Galimullin, and M.~Gladyshev.
\newblock Dynamic epistemic logic of resource bounded information mining
  agents.
\newblock In \emph{Proceedings of the 23rd International Conference on
  Autonomous Agents and Multiagent Systems}, AAMAS '24, page 481–489,
  Richland, SC, 2024. International Foundation for Autonomous Agents and
  Multiagent Systems.
\newblock ISBN 9798400704864.

\bibitem[Dudek et~al.(1991)Dudek, Jenkin, Milios, and Wilkes]{DudekEtAl1991}
G.~Dudek, M.~Jenkin, E.~Milios, and D.~Wilkes.
\newblock Robotic exploration as graph construction.
\newblock \emph{IEEE Transactions on Robotics and Automation}, 7\penalty0
  (6):\penalty0 859--865, 1991.
\newblock \doi{10.1109/70.105395}.

\bibitem[Fagin et~al.(1995)Fagin, Halpern, Moses, and Vardi]{FaginEtAl1995}
R.~Fagin, J.~Y. Halpern, Y.~Moses, and M.~Y. Vardi.
\newblock \emph{{Reasoning About Knowledge}}.
\newblock MIT Press, 1995.
\newblock \doi{10.7551/mitpress/5803.001.0001}.

\bibitem[Fernández~González(2018)]{FernandezGonzalez2018}
S.~Fernández~González.
\newblock Generic models for topological evidence logics.
\newblock {MS}c.~{T}hesis, ILLC, University of Amsterdam, Amsterdam, 2018.

\bibitem[Galmiche et~al.(2019)Galmiche, Kimmel, and Pym]{GalmicheEtAl2019}
D.~Galmiche, P.~Kimmel, and D.~Pym.
\newblock A substructural epistemic resource logic: theory and modelling
  applications.
\newblock \emph{Journal of Logic and Computation}, 29\penalty0 (8):\penalty0
  1251--1287, Dec. 2019.
\newblock ISSN 1465-363X.
\newblock \doi{10.1093/logcom/exz024}.

\bibitem[Green et~al.(2007)Green, Karvounarakis, and Tannen]{GreenEtAl2007}
T.~J. Green, G.~Karvounarakis, and V.~Tannen.
\newblock Provenance semirings.
\newblock In \emph{PODS 2007}, pages 31--40. ACM, 2007.
\newblock \doi{10.1145/1265530.1265535}.

\bibitem[Heifetz and Mongin(2001)]{HeifetzMongin2001}
A.~Heifetz and P.~Mongin.
\newblock Probability logic for type spaces.
\newblock \emph{Games and Economic Behavior}, 35\penalty0 (1–2):\penalty0
  31--53, Apr. 2001.
\newblock ISSN 0899-8256.
\newblock \doi{10.1006/game.1999.0788}.

\bibitem[Kozen et~al.(2013)Kozen, Mardare, and Panangaden]{KozenEtAl2013}
D.~Kozen, R.~Mardare, and P.~Panangaden.
\newblock Strong completeness for markovian logics.
\newblock In K.~Chatterjee and J.~Sgall, editors, \emph{Mathematical
  Foundations of Computer Science 2013. MFCS 2013.}, number 8087 in Lecture
  Notes in Computer Science, pages 655--666, Berlin, Heidelberg, 2013.
  Springer.
\newblock ISBN 9783642403132.
\newblock \doi{10.1007/978-3-642-40313-2_58}.

\bibitem[Kremer and Mints(2005)]{KremerMints2005}
P.~Kremer and G.~Mints.
\newblock Dynamic topological logic.
\newblock \emph{Annals of Pure and Applied Logic}, 131\penalty0
  (1–3):\penalty0 133--158, Jan. 2005.
\newblock ISSN 0168-0072.
\newblock \doi{10.1016/j.apal.2004.06.004}.

\bibitem[Kuich and Salomaa(1986)]{KuichSalomaa1986}
W.~Kuich and A.~Salomaa.
\newblock \emph{{S}emirings, {A}utomata, {L}anguages}.
\newblock EATCS Monographs on Theoretical Computer Science 5. Springer, 1986.
\newblock ISBN 978-3-642-69961-0,978-3-642-69959-7.
\newblock \doi{10.1007/978-3-642-69959-7}.

\bibitem[McKinsey and Tarski(1944)]{McKinseyTarski1944}
J.~C.~C. McKinsey and A.~Tarski.
\newblock The algebra of topology.
\newblock \emph{The Annals of Mathematics}, 45\penalty0 (1):\penalty0 141, Jan.
  1944.
\newblock ISSN 0003-486X.
\newblock \doi{10.2307/1969080}.

\bibitem[Munkres(2000)]{Munkres2000}
J.~R. Munkres.
\newblock \emph{Topology}.
\newblock Prentice Hall, 2nd edition edition, 2000.

\bibitem[Naumov and Tao(2015)]{NaumovTao2015}
P.~Naumov and J.~Tao.
\newblock Budget-constrained knowledge in multiagent systems.
\newblock In \emph{Proceedings of the 2015 International Conference on
  Autonomous Agents and Multiagent Systems}, AAMAS '15, page 219–226,
  Richland, SC, 2015. International Foundation for Autonomous Agents and
  Multiagent Systems.
\newblock ISBN 9781450334136.

\bibitem[Ramírez~Abarca(2015)]{RamirezAbarca2015}
A.~I. Ramírez~Abarca.
\newblock Topological models for group knowledge and belief.
\newblock Msc thesis, ILLC, University of Amsterdam, 2015.
\newblock URL \url{https://eprints.illc.uva.nl/id/eprint/2250}.

\bibitem[Sandhu et~al.(1996)Sandhu, Coyne, Feinstein, and
  Youman]{SandhuEtAl1996}
R.~Sandhu, E.~Coyne, H.~Feinstein, and C.~Youman.
\newblock Role-based access control models.
\newblock \emph{Computer}, 29\penalty0 (2):\penalty0 38--47, 1996.
\newblock ISSN 0018-9162.
\newblock \doi{10.1109/2.485845}.

\bibitem[Smyth(1992)]{Smyth1992}
M.~B. Smyth.
\newblock Topology.
\newblock In S.~Abramsky, D.~M. Gabbay, and T.~S.~E. Maibaum, editors,
  \emph{Handbook of Logic in Computer Science}, volume~1, pages 641--761.
  Oxford University Press, Oxford, 1992.
\newblock ISBN 9781383026023.
\newblock \doi{10.1093/oso/9780198537359.003.0005}.

\bibitem[Solaki(2022)]{Solaki2022}
A.~Solaki.
\newblock Actualizing distributed knowledge in bounded groups.
\newblock \emph{Journal of Logic and Computation}, 33\penalty0 (6):\penalty0
  1497--1525, Mar. 2022.
\newblock ISSN 1465-363X.
\newblock \doi{10.1093/logcom/exac007}.

\bibitem[{ten Cate} et~al.(2009){ten Cate}, Gabelaia, and
  Sustretov]{tenCateEtAl2009}
B.~{ten Cate}, D.~Gabelaia, and D.~Sustretov.
\newblock Modal languages for topology: {E}xpressivity and definability.
\newblock \emph{Annals of Pure and Applied Logic}, 159\penalty0
  (1–2):\penalty0 146--170, May 2009.
\newblock ISSN 0168-0072.
\newblock \doi{10.1016/j.apal.2008.11.001}.

\bibitem[{van Benthem} et~al.(2014){van Benthem}, {Fernández-Duque}, and
  Pacuit]{vanBenthemEtAl2014}
J.~{van Benthem}, D.~{Fernández-Duque}, and E.~Pacuit.
\newblock {Evidence and plausibility in neighborhood structures}.
\newblock \emph{Annals of Pure and Applied Logic}, 165\penalty0 (1):\penalty0
  106--133, Jan. 2014.
\newblock ISSN 0168-0072.
\newblock \doi{10.1016/j.apal.2013.07.007}.

\bibitem[Vickers(1989)]{Vickers1989}
S.~Vickers.
\newblock \emph{Topology via Logic}.
\newblock Cambridge University Press, 1989.

\bibitem[Zhou(2007)]{Zhou2007}
C.~Zhou.
\newblock \emph{{Complete Deductive Systems for Probability Logic with
  Application to Harsanyi Type Spaces}}.
\newblock Phd thesis, Indiana University, 2007.

\bibitem[Zhou(2014)]{Zhou2014}
C.~Zhou.
\newblock Probability logic for harsanyi type spaces.
\newblock \emph{Logical Methods in Computer Science}, 10\penalty0 (2), 2014.
\newblock ISSN 1860-5974.
\newblock \doi{10.2168/lmcs-10(2:13)2014}.

\bibitem[Özgün et~al.(2025)Özgün, Smets, and Zotescu]{OzgunEtAl2025}
A.~Özgün, S.~Smets, and T.-S. Zotescu.
\newblock Evidence diffusion in social networks: a topological perspective.
\newblock In V.~Goranko, C.~Shi, and W.~Wang, editors, \emph{Logic,
  Rationality, and Interaction}, pages 110--123. Springer Nature Singapore,
  Oct. 2025.
\newblock ISBN 9789819524815.
\newblock \doi{10.1007/978-981-95-2481-5_8}.

\end{thebibliography}

\appendix
\section{Appendix: Full proofs and details}\label{app}

This appendix contains the full proofs of the statements made in the main text.

\subsection{Failure of interior properties for $\Int_a$}\label{app:Inta}
 
We show that the properties (i) $\Int_a(X) = X$, (ii) $Int_a(P) \subseteq \Int_a\Int_a(P)$ and (iii) $\Int_a(P) \cap \Int_a(Q) \subseteq \Int_a(P \cap Q)$ do not hold in all seats $\langle X, \topo, \ann_K\rangle$, providing a counter-example with $K = \mathbb{N}^\infty = \langle \mathbb{N} \cup \{ \infty \}, \min, +, \0, \1 \rangle$.  

Let $X = \{ x,y,z \}$ be a three-element set with topology $\topo = \{ \varnothing, \{ x \}, \{ x,y \}, \{ x,z \}, X  \}$. 
Let $\ann(X,x) = \ann(\{x,y\}, x) = \ann(\{ x,z \}, x) = {\uparrow}42$ and let $\ann$ return ${\uparrow}43$ in all remaining cases. Then $\Int_{42}(X) = \{ x \}$ contra (i), $\Int_{42}(\{x\}) = \varnothing$ contra (ii) and $\Int_{42}(\{ x,y \}) \cap \Int_{42}(\{ x,z \}) = \{ x \}$ contra (iii).

\subsection{Proof of Lemma \ref{lem:SoudnessS4K}}\label{app:SoudnessS4K}
We prove that all axioms are valid and that all inference rules preserve validity in all $K$-seats. Since $\llbracket \Box\f\rrbracket = \Int \llbracket \f\rrbracket$, validity of the $\mathbf{S4}$ axioms
\[ \Box (\f \land \ff) \tot (\Box \f \land \Box \ff) \qquad \Box \f \to \f \qquad \Box \f \to \Box\Box \f\] 
directly follows from the corresponding properties of the interior operator: $\Int(P \cap Q) = \Int(P) \cap \Int(Q)$, $\Int(P) \subseteq P$ and $\Int(P) \subseteq \Int(\Int(P))$, respectively.  
The rule $\f / \Box \f$ preserves validity in $K$-seats because $\Int(X) = X$ in each topological space. Validity of the additional $\mathbf{S4}_K$ axioms is established as follows.

\eqref{a:Fmult}: If $\bm{M}, x \models \F_a\f$, then there is an open $U \subseteq \semM{\f}$ with $a \in \ann_x(U)$. Since $ab, ba \in \ann_x(U)$ by \eqref{A(i)}, we immediately obtain $\bm{M}, x \models \F_{ab}\f \wedge \F_{ba}\f$.

\eqref{a:Fplus}: Suppose $\bm{M}, x \models \F_a\f \wedge \F_b\ff$. Then there are opens $U \subseteq \sem{\f}_{\bm{M}}$ and $V \subseteq \sem{\ff}_{\bm{M}}$ with $a \in \ann_x(U)$ and $b \in \ann_x(V)$. By \eqref{A(ii)} we get $a,b \in \ann_x(U\cup V)$ and by \eqref{A(iii)} this yields $a\oplus b \in \ann_x(U \cup V)$. Since $U \cup V \subseteq \Int(\sem{\f}_{\bm{M}}) \cup \Int(\sem{\ff}_{\bm{M}}) = \sem{\Box\f \vee \Box \ff }_{\bm{M}}$, we have $\bm{M}, x \models \F_{a\oplus b}(\Box\f \vee \Box\ff)$.

\eqref{a:Fab}: Suppose $\bm{M}, x \models \F_a\f \wedge \F_b\ff$. Then there are opens $U \subseteq \sem{\f}_{\bm{M}}$ and $V \subseteq \sem{\ff}_{\bm{M}}$ with $a \in \ann_x(U)$ and $b \in \ann_x(V)$. By \eqref{A(iv)} we have $ab \in \ann_x(U \cap V)$. Since $U \cap V \subseteq \semM{\f} \cap \semM{\ff} = \semM{\f \wedge \ff}$, we obtain $\bm{M}, x \models \F_{ab}(\f \wedge \ff)$. 

\eqref{a:Top}: Thanks to \eqref{A(v)} we have $\0\in \ann_x(X)$ and together with $\semM{\top} = X$ this gives us $\bm{M}, x \models \F_\0\top$.   

Lastly, \eqref{a:FBox} and \eqref{a:Fmono} follow directly from Lemma~\ref{lem:SoundnessHelp} (items 3.\ and 1., respectively). \qed

\subsection{Proof of Lemma \ref{lem:CanModelS4K}}\label{app:CanModelS4K}
It is sufficient to prove that the canonical $\langle X, \topo, \ann_K\rangle$ is a seat, that is, to show \eqref{A(i)}--\eqref{A(v)} of Definition~\ref{def:seat}.           

\eqref{A(i)}:
Let $a \in \ann_\Gamma(U)$ and $b \in K$. By definition of the former there is $\F_a\f \in \Gamma$ with $|\Box \f| \subseteq U$ and using \eqref{a:Fmult} we obtain $\F_{ab}\f, \F_{ba}\f\in \Gamma$ witnessing $ab, ba \in \ann_\Gamma(U)$ as desired.

\eqref{A(ii)}:
It is straightforward to see by definition that $a \in \ann_\Gamma(U)$ and $U \subseteq V$ imply $a \in \ann_\Gamma(V)$.

\eqref{A(iii)}:
If $a,b \in \ann_\Gamma(U)$, then there are $\F_a\f, \F_b\ff \in \Gamma$ with $|\Box\f| \subseteq U$ and $|\Box\ff|\subseteq  U$. Therefore, $|\Box \f \vee \Box\ff| \subseteq U$ and furthermore $\F_{a\oplus b}(\Box\f \vee \Box \psi) \in \Gamma$ due to \eqref{a:Fplus}, which witnesses $a\oplus b \in \ann_\Gamma(U)$ together with the axiom $\Box \chi \to \chi$ of $\mathbf{S4}$.

\eqref{A(iv)}:
If $a \in \ann_\Gamma(U)$ and $b \in \ann_\Gamma(V)$, then there are formulas $\f, \ff$ such that $\F_a \f \land \F_b \ff \in \Gamma$, $| \Box \f| \subseteq U$ and $| \Box \ff| \subseteq V$. Hence, $|\Box (\f \land \ff)| = |\Box \f| \cap |\Box \ff| \subseteq U \cap V$. By \eqref{a:Fab}, $\F_{ab}(\f \land \ff) \in \Gamma$. 

\eqref{A(v)}:
Due to \eqref{a:Top} we have $\F_\0\top \in \Gamma$, and together with $|\Box \top| = |\top| = X$ this shows $\0 \in \ann_\Gamma(X)$ as desired. 
\qed

\subsection{Proof of Lemma \ref{lem:TruthLemmaS4K}}\label{app:TruthLemmaS4K}
Structural induction on $\chi$. The base case $\chi \in \Prop$ holds by definition and the inductive  cases for $\chi = \neg \f$ and $\chi = \f \land \f'$ are established using the induction hypothesis $|\ff| = \llbracket \ff\rrbracket_{\bm{M}^{\mathbf{S4}_K}}$ for $\ff \in \{ \f, \f' \}$ and the properties of maximal consistent theories in the usual way.

The case $\chi = \Box \f$ is established as follows. If $\Box \f \in \Gamma$, then $\Gamma \in | \Box \f | \in \topo$ and $| \Box \f | \subseteq |\f|$ using the $\mathbf{S4}$-axiom $\Box \f \to \f$. 
Conversely, if $U \in \topo$ and $\Gamma \in U$, then there is $\Box \ff$ such that $\Box \ff \in \Gamma$ (since $U = \bigcup_{i \in I} | \Box \ff_i |$). If also $U \subseteq |\f|$, then $|\Box \ff | \subseteq |\f|$, which means that $\Box \ff \to \f$ is provable, whence $\Box \ff \to \Box \f$ is provable by $\mathbf{S4}$ and so $\Box \f \in \Gamma$. 

The case $\chi = \F_a \f$ is established as follows. If $\F_a \f \in \Gamma$, then $|\Box \f| \in \evid_a(\Gamma)$ and $|\Box \f| \subseteq |\f|$ by $\mathbf{S4}$. Conversely, assume that there is $U \in \topo$ such that $\evid_a(\Gamma)$ and $U \subseteq |\f|$; $\evid_a(\Gamma)$ means that $\exists \ff$ such that $\F_a \ff \in \Gamma$ and $|\Box \ff| \subseteq U$. It follows that $|\Box \ff| \subseteq | \f |$, which means that $\Box \ff \to \f$ is provable in $\mathbf{S4}_K$. Hence, $\F_a \Box \ff \to \F_a \f$ is provable by \eqref{a:Fmono}, which means that $\F_a \ff \to \F_a \f$ is provable by \eqref{a:FBox}. Hence, $\F_a \f \in \Gamma$. \qed

\subsection{Proof of Lemma~\ref{lem:CharactS4sb}}\label{app:charactS4sb}
Suppose $\langle X, \topo, \ann \rangle$ is a strong bounded seat. We show that it validates \eqref{a:Fplus2} -- \eqref{a:Bot0}. Let $\bm{M} = \langle X, \topo, \ann_K, \val \rangle$ be a $K$-model and $x \in X$. If $\bm{M}, x \models \F_{a\oplus b}\f$, then $\exists U \in \topo$ with $a\oplus b \in \ann(U,x)$ and $U \subseteq \semM{\f}$. By definition of strong seats, this implies $a,b \in \ann(U,x)$ which yields $\bm{M}, x \models \F_{a}\f \wedge \F_b\f$, establishing \eqref{a:Fplus2}. 
For \eqref{a:Top1} -- \eqref{a:Bot0}, boundedness entails $\1 \in \ann(X,x)$ and $\0 \in \ann(\varnothing, x)$, which together with $X \subseteq \semM{\top}$ and $\varnothing \subseteq \semM{\bot}$ yields $\bm{M},x \models \F_\1\top$ and $\bm{M},x \models \F_\0\bot$.

Conversely, assume that $\langle X, \topo, \ann \rangle$ is not strong bounded. If it is not strong, then there are $a,b \in K$, $U \in \topo$ and $x \in x$ such that $a\oplus b \in \ann(U,x)$ but $a \notin \ann(U,x)$. Let $\bm{M} = \langle X, \topo, \ann, \val \rangle$ be any model with $\val(p) = U$. Then $\bm{M}, x \models \F_{a\oplus b} p$ but $\bm{M},x \not\models \F_a p$, showing that \eqref{a:Fplus2} is not valid. 
If it is not $\1$-bounded, there is $x\in X$ with $\1 \notin \ann(X,x)$. Then $\bm{M}, x \not \models \F_\1\top$ in any model, since otherwise there is $U\in \topo$ with $\1 \in \ann(U, x)$ and by \eqref{A(ii)} $\1 \in \ann(X,x)$, a contradiction. If it is not $\0$-bounded, there is $x \in X$ with $\0 \notin \ann(\varnothing, x)$. Then clearly $\bm{M}, x \not \models \F_\0\bot$. \qed

\subsection{Proof of Lemma~\ref{lem:useful}}\label{app:useful}
Assume $|\f|^{\Lambda} \subseteq |\ff|$. By compactness of $\proves_{\mathbf{S4sub}_K}$ we infer
\[\proves_{\mathbf{S4sub}_K} \f \land \bigwedge_{i = 1}^{n} \pm \F_{a_i} \chi_i \to \ff\] for some finite $\{ \pm \F_{a_i} \chi_i\}_{i =1}^n \subseteq \Lambda$. In the following, let $\F^\ast_a = \F_\1$ and $\Box^\ast = \Box$. We continue as follows:
\begin{align*}
& \proves_{\mathbf{S4sub}_K} \M \left (\f \land \bigwedge_{i = 1}^{n} \pm \F_{a_i} \chi_i \right ) \to \M\ff && \text{by \eqref{a:Fmono}, $\mathbf{S4}$}\\[2mm]
& \proves_{\mathbf{S4sub}_K} \M \f \land \bigwedge_{i = 1}^{n} \M^\ast \mathord{\pm} \F_{a_i} \chi_i \to \M \ff && \text{by \eqref{a:Fab}, $\mathbf{S4}$}\\[2mm]
& \proves_{\mathbf{S4sub}_K} \M \f \land \bigwedge_{i = 1}^{n} \mathord{\pm} \F_{a_i} \chi_i \to \M \ff && \text{by \eqref{au:0}, \eqref{au:Box}}
\end{align*}
Hence, for any $\Gamma \in X^\Lambda$, if $\Gamma \in |\M\f| \cap X^{\Lambda}$, then $\Gamma \in |\M\ff| \cap X^\Lambda$ as desired.  \qed

\subsection{Proof of Lemma \ref{lem:CanModelS4uK}}\label{app:CanModelS4uK}

The fact that $\bm{M}^{\Lambda}$ satisfies conditions \eqref{A(i)} through \eqref{A(v)} follows from the fact that $\bm{M}^{\mathbf{S4sub}_K}$ satisfies them. Take \eqref{A(iii)}, for instance. The assumption $a, b \in \ann^{\Lambda}(U)$ for $U \in \topo^{\Lambda}$ means that $a \in \ann (U_a, \Lambda)$ and $b \in \ann (U_b, \Lambda)$ for some $U_a, U_b \in \topo$ such that $(U_a \cup U_b) \cap X^{\Lambda} \subseteq U$ (Lemma \ref{lem:CompletS4uK-help}). We know that $a \oplus b \in \ann (U_a \cup U_b, \Lambda)$, and so $a \oplus b \in \ann^{\Lambda}(U)$ as desired. 

In a similar vein, the fact that $\bm{M}^{\Lambda}$ is strong and bounded follows from the fact that $\bm{M}^{\mathbf{S4sub}_K}$ is. To show that it is strong, assume that $a \oplus b \in \ann^{\Lambda}(U)$ for some $U \in \topo^{\Lambda}$. Hence, $a \oplus b \in \ann(V, \Lambda)$ for some $V \in \topo$ such that $V \cap X^{\Lambda} \subseteq U$. But we know that $a, b \in \ann (V, \Lambda)$, whence $a, b \in \ann^{\Lambda}(U)$. 
Boundedness is established in a similar fashion. \qed

\subsection{Proof of Lemma~\ref{lem:useful global}}\label{app:useful global}
Assume $|\f|^{\Lambda} \subseteq |\ff|$. By compactness, we infer
\[\proves_{\mathbf{S4sb}_{K\forall}} \f \land \bigwedge_{i = 1}^{n}   \chi_i \to \ff\] for some finite $\{ \A \chi_i\}_{i=1}^n \subseteq \Lambda$. 

In the following, let $\F_a^\ast = \F_\1$, $\Box^\ast = \Box$ and $\A^\ast = \A$. We continue as follows:
\begin{align*}
& \proves_{\mathbf{S4sb}_{K\forall}} \M \left (\f \land \bigwedge_{i = 1}^{n} \chi_r \right ) \to \M \ff && \text{by montonicity of $\M$}\\[2mm]
& \proves_{\mathbf{S4sb}_{K\forall}} \M \f \land \bigwedge_{i = 1}^{n} \M^\ast  \chi_r \to \M \ff && \text{by $\mathbf{S4}$, \eqref{a:Fab}, \eqref{ax:Anorm}}\\
& \proves_{\mathbf{S4sb}_{K\forall}} \M \f \land \bigwedge_{i = 1}^{n} \mathord \A \chi_r \to \M \ff && \text{by \eqref{ax:ABox}}
\end{align*}
Since, $\A \chi_i \in \Lambda$ for all $i \leq n$, by \eqref{ax:Atrans}, $\A\A \chi_i \in \Lambda$. Hence, by definition, for all $\Gamma \in X^\Lambda$, $\A\chi_i \in \Gamma$ for all $i \leq n$.  Therefore, if  $\Gamma \in |\M\f| \cap X^{\Lambda}$, then $\Gamma \in |\M\ff|\cap X^\Lambda$, as desired. \qed

\subsection{Proof of Lemma~\ref{lem:truth S4sbKforall}}\label{app:TruthLemmaAll}
The proof is by induction on the structure of $\chi$, where $\chi = p \in \Prop$ and the inductive cases of propositional connectives are straightforward. The inductive cases of $\Box$ and $\F_a$ are analogous to the corresponding proofs in Lemma \ref{lem:TruthLemmaS4uK}, except for using Lemma \ref{lem:useful global} in place of Lemma \ref{lem:useful}.

\emph{Case} $\chi =\A \f$: Let $\Gamma \in X^\Lambda$.  If $\A\f \in \Gamma$, then by \eqref{ax:Atrans}, $\A\A\f \in \Gamma$. By construction $X^\Lambda \subseteq |\A\f|$. Therefore, $\Gamma \in \llbracket \A\f \rrbracket_{\bm{M}^\Lambda}$. 

For the converse, suppose $\A\f \notin \Gamma$ but $\Gamma \in \llbracket \A\f \rrbracket_{\bm{M}^\Lambda}$. Then, $|\top|^\Lambda = X^\Lambda \subseteq |\f|^\Lambda$.  Note that by global necessitation, $|\A\top| =X$ which implies $|\A\top|^\Lambda =X^\Lambda$. Thus, by Lemma \ref{lem:useful global}, $|\A\top|^\Lambda =X^\Lambda \subseteq |\A\f|^\Lambda$. Therefore, $\Gamma \in |\A\f|^\Lambda$, a contradiction. \qed

\subsection{Proof of Theorem\ref{thm:characterize uniform}}\label{app:characterize uniform}
The soundness follows from Lemma \ref{lem:uniform define}. We prove completeness by the canonical model construction, analogous to the canonical model construction for $\mathbf{S4sb}_{K\forall}$, except for defining $X$ to be the set of all maximal $\mathbf{S4sub}_{K\forall}$ consistent theories. It is enough to show that the canonical model $\bm{M}^\Lambda$
constructed in this manner is a strong uniform bounded $K$-model. It is shown exactly as in the proof of Lemmas~\ref{lem:CanModelS4K} \& ~\ref{lem:CanModelS4uK} that this is a strong bounded $K$-model. By the  proof of Lemma \ref{lem:CompletS4uK-help}, for uniformity it suffices to show that $\F_a \f \in \Gamma$ iff $\F_a \f \in \Lambda$ for all $\Gamma \in X^\Lambda$, $\f \in  \mathfrak{L}_{K\forall}$, and $a \in K$. By axiom \eqref{ax:uniformity global}, if $\F_a \f \in \Lambda$, then $\A\F_a \f \in \Lambda$. Therefore, by construction $X^\Lambda \subseteq |\F_a \f|$, which implies $\F_a\f \in \Gamma$. Conversely, if $\F_a \f\not \in \Lambda$, then $\neg \F_a \f \in \Lambda$. By axiom \eqref{ax:uniformity global 2}, $\A\neg \F_a \f \in \Lambda$.  Therefore, by construction $X^\Lambda \subseteq |\neg \F_a \f|$, which implies $\F_a\f \not \in \Gamma$.\qed

\subsection{Proof of Lemma \ref{lem:belief knowledge} \label{app:belief knowledge}}
We only prove the result for $\Belop^a_b $. The result for $\Knop^a_b$ follows from it in a straightforward manner.  Since $\bm{M}$ is uniform, we use $\evid_a$ to denote the set $\evid_a(x)$ for any $x \in X$. 

\begin{tabular}{lll}
     & $\bm{M}, x \models \Bel^a_b \f$ & \\
     iff &  $\bm{M}, x \models \A\Diamond_b\Box_a \f $ & (Def. of~$\Bel^a_b $) \\
     iff & for all $x' \in X$, $\bm{M}, x' \models \Diamond_b\Box_a \f $ & (Def. of~$\A$)\\
     iff & for all $x' \in X$, and $U \in \evid_b$ 
     & (Def.~of $\Diamond_b$)\\ &  s.t.~if $x' \in U \in \evid_b$, then there  \\  
      &  exists 
      $x'' \in U$ s.t.~$\bm{M}, x'' \models \Box_a \f$& \\

      iff & for all  $U$ s.t. $\varnothing \neq U \in \evid_b$ \\ &
      there    exists 
      $x'' \in U$  \\&
      s.t.~$\bm{M}, x'' \models \Box_a \f$ \\
       iff & for all $U$ s.t. $\varnothing \neq U   \in \evid_b$, & (Def.~of $\Box_a$) \\ &  there  exist 
      $x'' \in U$   
      and  \\ & $U'\in \evid_a$
      s.t.~$x'' \in U' \subseteq |\f|$ \\
    iff & for all $U$ s.t. $\varnothing \neq U   \in \evid_b$, \\
      &  there exists $U' \in \evid_a$ s.t. \\
      &  $U \cap U' \neq \varnothing$ and   $U' \subseteq |\f|$.\\
\end{tabular}

Let $V = \bigcup_{\varnothing \neq U   \in \evid_a}U'$, where  for every $U$, $U'$ is as in the last statement of the above iff chain.  Then, we have $V \subseteq |\f|$, $V \in \evid_a$ and $V \cap U \neq \varnothing$ for all
non-empty sets  $U  \in \evid_b$. Therefore, $V$ is $b$-dense. Thus, by definition, $x \in  \Belop^a_b (\semM{\f})$.

Conversely, if  $x \in  \Belop^a_b (\semM{\f})$, then there exists $D \in \evid_a$ s.t.~$D \subseteq \semM{\f}$ and for all $U \in \evid_b$, $U \neq \varnothing$ implies  $D \cap U \neq \varnothing$.  Then, choosing $U'=D$ for any $U$ makes the last statement of the above iff chain true. Consequently, $\bm{M}, x \models \Bel^a_b \f$. \qed

\subsection{Proof of Proposition~\ref{prop:DisjointUnion}}
\label{app:DisjointUnion}
The proof is by induction on $\chi$, where the base  case $p \in\Prop$ and the inductive cases for propositional connectives are immediate. The case $\chi = \Box\f$ follows from $\Int^\topo(\semM{\f}) \cap X_i = \Int^{\topo_i}(\semM{\f} \cap X_i)$ and the induction hypothesis.

For the remaining case $\chi = \F_a\f$, take $x \in X_i$. 
If $x \in \sem{\F_a\f}_{\bm{M}_i}$, then there is $U_i \in \topo_i$ with $a \in \ann_i(U_i, x)$ and $U_i \subseteq \sem{\f}_{\bm{M}_i}$. Considering (up to inclusion) $U_i \in \topo$, we thus get $a \in \ann_i(U_i, x) = \ann(U_i, x)$, and by induction hypothesis, $U_i \subseteq \sem{\f}_{\bm{M}_i} = \sem{\f}_{\bm{M}} \cap X_i \subseteq \sem{\f}_{\bm{M}}$ shows $x \in \sem{\F_a\f}_{\bm{M}}$. 
Conversely, if $x \in \sem{\F_a\f}_{\bm{M}}$, then there exists some $U \in \topo$ with $a \in \ann_i(U, x)$ and $U \subseteq \sem{\f}_{\bm{M}}$. For $U \cap X_i \in \topo_i$ we have $a \in \ann(U, x) = \ann_i(U\cap X_i, x)$, and by the induction hypothesis, we have $U \cap X_i \subseteq \semM{\f} \cap X_i = \sem{\f}_{\bm{M}_i}$. Thus,  $x \in \sem{\F_a\f}_{\bm{M}_i}$. \qed

\subsection{Proof of Theorem \ref{thm:bisimulation}}
\label{app:Proof of Theorem bisimulation}

The proof is by induction on the complexity of the formula $\f$. The proof for the base case (propositional variables) and the inductive step for propositional connectives is trivial. We give the proof for modal operators $\Box$, $\F_a$, and $\A$ (in the case of global bisimulations).

(1) Suppose $\phi=\Box \psi$ for some formula $\psi$.
Then, $\bm{M}_1,  x_1 \models \Box \psi$ iff   $\exists U_1 \in \topo_1\colon x_1 \in U_1 \subseteq \llbracket \f\rrbracket_{\bm{M}_1}$. That is, $\exists U_1 \in \topo_1$ s.t.~$x_1 \in U_1$, and for all $y_1 \in U_1$,
 $\bm{M}_1,  y_1 \models  \psi$. By Definition~\ref{def:bisimulation}(ii), $\exists U_2 \in \topo_2$ s.t.~$x_2 \in U_2$ and for all $y_2 \in U_2$, there exists $y_1 \in U_1$ s.t.~$y_1 Z y_2$. By induction, this implies $\exists U_2 \in \topo_2$ s.t.~$x_2 \in U_2$ and for all $y_2 \in U_2$,  $\bm{M}_2,  y_2 \models  \psi$. Consequently, $\bm{M}_1,  y_1 \models \Box \psi$. The preservation of satisfaction in the opposite direction is shown analogously using Definition~\ref{def:bisimulation} (iii).

 (2) Suppose $\phi=\F_a \psi$ for some formula $\psi$.
Then, $\bm{M}_1,  x_1 \models \F_a \psi$ iff $\exists U_1 \in \topo_1\colon U_1 \subseteq \llbracket \f\rrbracket_{\bm{M}_1}$. That is, $\exists U_1 \in \topo_1$ s.t.~for all $y_1 \in U_1$
 $\bm{M}_1,  y_1 \models  \psi$. By Definition ~\ref{def:bisimulation} (iv), $\exists U_2 \in {\epsilon_a}_2(x_2)$ s.t.~for all $y_2 \in U_2$ there exists $y_1 \in U_1$ s.t.~$y_1 Z y_2$. By induction, this implies $\exists U_2 \in {\epsilon_a}_2(x_2)$ s.t.~for all $y_2 \in U_2$,  $\bm{M}_2,  y_2 \models  \psi$. Consequently, $\bm{M}_1,  y_1 \models \F_a \psi$. The preservation of satisfaction in the opposite direction is shown analogously using Definition~\ref{def:bisimulation} (v).

(3) Suppose $\phi=\A \psi$ for some formula $\psi$.
Then, $\bm{M}_1,  x_1 \models A \psi$ iff $\forall y_1 \in X_1\colon\bm{M}_1,  y_1 \models \psi$. By Definition ~\ref{def:bisimulation} (vii), for any point $y_2 \in X_2$, there exists $y_1 \in X_1$ s.t.~$y_1 Z y_2$. Since  $\bm{M}_1,  y_1 \models \psi$,  by induction 
$\bm{M}_2,  y_2 \models \psi$ for any $y_2 \in X_2$. Therefore,  $\bm{M}_2,  x_2 \models \A \psi$.  The preservation of satisfaction in the opposite direction is shown analogously using Definition~\ref{def:bisimulation} (vi).\qed

\subsection{Proof of Corollary \ref{cor:bisimul example}}\label{app:corr example}
 Consider the semiring  $K = \langle \mathbb{Q}_{\geq 0}^{\infty}, \min, +, \infty, 0 \rangle$. 

1. Let  $\bm{M}_1=\langle X_1,\topo_1,\ann_1,\allowbreak \val_1\rangle$ and $\bm{M}_2=\langle X_2,\topo_2,\ann_2,\val_2\rangle$ be  strong uniform bounded $K$-models defined as follows. Let $X_1=\{x_1\}$, $\topo_1=\{\varnothing, X_1\}$,
${\ann_1}(\varnothing,x_1)={\ann_1}(X_1,x_1)= \mathbb{Q}_{\geq 0}^{\infty}$, and for all $p \in \Prop$,  $\val_1(p)=X_1$ .
 Let $X_2=\{x_2,y_2\}$, $\topo_2=\{\varnothing, \{x_2\},  X_2\}$,
${\ann_2}(\varnothing,x)= {\ann_2}_K(\{x_2\},x)={\ann_2}(X_2,x)= \mathbb{Q}_{\geq 0}^{\infty}$ for all $x \in X_2$ and $\val_2(p)=\{x_2\}$ for all $p \in \Prop$.

Define $Z=\{(x_1,x_2)\}$.
It is straightforward to check that $Z$ is a bisimulation.
However, $\bm{M}_1,x_1 \models \A p$ but $\bm{M}_2,x_2 \not\models \A p$. Therefore, $\A$ is not definable by any formula in $\mathfrak{L}_{K}$.

2. Let  $\bm{M}_1=\langle X_1,\topo_1,\ann_1,\allowbreak \val_1\rangle$ and $\bm{M}_2=\langle X_2,\topo_2,\ann_2,\val_2\rangle$ be  strong uniform bounded $K$-models  defined as follows. Let $X_1=\{x_1,y_1,z_1\}$, $\topo_1=\{\varnothing, \allowbreak \{x_1,y_1\},\allowbreak X_1\}$, and for all $p \in \Prop$, $\val_1(p)=X_1$ . For all $x\in X_1$, ${\ann_1}(X_1,x)=[0,\infty]$, ${\ann_1}(\{x_1, y_1\},x)=(0,\infty]$,  and 
${\ann_1}(\varnothing,x)=\{1\}$.  Let $X_2=\{x_2,y_2\}$, $\topo_2=\{\varnothing, X_2\}$,  and  for all $p \in \Prop$, $\val_2(p)=X_2$. For all $x \in X_2$,
${\ann_2}(X_2,x)=[0,\infty]$, and ${\ann_2}(\varnothing,x)=\{\infty\}$. 

Set $Z=\{(x_1,x_2),(y_1,y_2),(z_1,y_2)\}$.
It is straightforward to check that $Z$ is a global bisimulation. However,  $\bigsqcup  {\ann_1}(U, x_1) \not \in  {\ann_1}(U,x_1)$ for  $U=\{x_1,y_1\}$,  but $\bigsqcup  {\ann_2}(U,x_2) \in {\ann_2}(U,x_2)$  for all $U \in \topo_2$.  Therefore, by Theorem~\ref{thm:bisimulation},  $\forall U \in \topo\colon\bigsqcup \ann_x(U) \in \ann_x(U)$ is not (locally) definable in $\mathbf{S4sub}_{K\forall}$. \qed

\end{document}